\documentclass[12pt]{article}
\usepackage{bm,amsmath,amssymb}
\usepackage{graphicx}
\usepackage{natbib}
\usepackage{url} % not crucial - just used below for the URL 
\usepackage{color}
\usepackage{hyperref}
\usepackage{amsthm}
\usepackage{todonotes}
\usepackage{verbatim}
\usepackage{enumitem}
\usepackage{thm-restate}
\usepackage{algorithm}
\usepackage{algorithmic}    
\usepackage{nicefrac}
\usepackage{booktabs}
\usepackage{array}
\usepackage{threeparttable}
\usepackage{multirow}
\usepackage{siunitx}
\usepackage{subcaption} % In preamble

\usepackage{xr}

\newcommand{\blind}{0}

\newcommand{\mscr}[1]{\mathcal{#1}}

\newcommand{\twid}[1]{\widetilde{#1}}

\newcommand{\ZZ}{\mathbb{Z}}
\newcommand{\RR}{\mathbb{R}}

\DeclareMathOperator{\EE}{\mathbb{E}}

\newcommand{\ep}{\epsilon}

\newcommand{\sdp}{s_{dp}}

\newcommand{\iid}{\overset{\text{iid}}{\sim}}

\DeclareMathOperator{\var}{Var}

\newcommand{\ol}{\overline}

\newtheorem{thm}{Theorem}%[section] % reset theorem numbering for each section

\newtheorem{cor}[thm]{Corollary}
\newtheorem{lemma}[thm]{Lemma}

\newtheorem{example}[thm]{Example}
\newtheorem{remark}[thm]{Remark}

\newtheorem{assumpD}{Assumption}
\newtheorem{assumpM}{Assumption}

\usepackage{mathtools}
\newcommand{\defeq}{\vcentcolon=}

\newcommand{\efdeq}{=\vcentcolon}
\begin{document}

\def\spacingset#1{\renewcommand{\baselinestretch}%
{#1}\small\normalsize} \spacingset{1}

%%%%%%%%%%%%%%%%%%%%%%%%%%%%%%%%%%%%%%%%%%%%%%%%%%%%%%%%%%%%%%%%%%%%%%%%%%%%%%

\if0\blind
{
  \title{\bf Large-Sample Bayesian Approximations\\ for Privatized Data}
  \author{Jordan Awan\thanks{
    Jordan Awan and Roberto Molinari gratefully acknowledge support from NSF grant SES-2150615.} \thanks{Jordan Awan's work was also partially supported by NSF grant SES-2610910}\hspace{.2cm}\\
    Department of Statistics,  University of Pittsburgh\\
   \\ 
    %and \\
    Xi Chen\\
    Department of Politics and International Affairs, Northern Arizona University\\
    \\
    %and\\
    Roberto Molinari$^*$\\
    Department of Mathematics and Statistics, Auburn University}
    \date{}
  \maketitle
} \fi

\if1\blind
{
  \bigskip
  \bigskip
  \bigskip
  \begin{center}
    {\LARGE\bf Large-Sample Bayesian Approximations\\ for Privatized Data}
\end{center}
  \medskip
} \fi

\bigskip
\begin{abstract}
The increased use of differential privacy (DP)  has allowed the sharing of large amounts of data while reducing the risk of disclosure of sensitive information at the individual level. However, the noise introduced by DP methods makes performing statistical inference more challenging. While various methods have been proposed to address different inferential tasks, they often require strong parametric assumptions and/or do not scale well with sample sizes (e.g. U.S. Census products).
%The trade-off for privacy however is a reduction in the utility of the data itself given that differential privacy introduces calibrated noise into the statistics thereby negatively affecting the following inferential process. To address this, many methods have been put forward over recent years to deliver statistically reliable conclusions for specific analytical approaches and certain differentially private mechanisms. Nevertheless, in many cases these solutions are either computationally demanding and/or rely on strong distributional assumptions for the data. 
In response to these limitations, we propose an approximate Bayesian method to analyze privatized data products, which uses a two-step approach of imputing the confidential data and then sampling from the non-private posterior, and which is inspired by the method of \citet{guha2024causal}. %Unfortunately, as originally proposed, the method requires strong assumptions and lacks concrete statistical guarantees. 
We prove that this approximate sampler is asymptotically valid under mild assumptions. While this approach is motivated by Bayesian theory, we show through simulations that it provides conservative frequentist properties as well. We demonstrate the utility of our method by applying it in simulated settings as well as for an analysis on the drivers of homeownership via the 2022 American Community Survey. 
\end{abstract}

\noindent%
\vfill

\newpage
\spacingset{1.75} % DON'T change the spacing!

\section{Introduction}
\label{s:intro}

Differential privacy (DP), introduced in \citet{dwork2006calibrating}, is currently the state-of-the-art framework for formal privacy protection. To protect privacy, DP methods require the introduction of external randomness into the data product to obscure the effect of any one individual's data on the result. DP is now employed by technological companies such as Google \citep{erlingsson2014rappor}, Apple \citep{tang2017privacy}, and Facebook \citep{evans2023statistically}, as well as by the U.S. Census Bureau \citep{abowd2019census}. In particular, products from the U.S. Census are now subject to DP methods to reduce the risk of individuals being identified when researchers run their analysis on this data. 
This approach however negatively affects the work of political and socio-economic researchers who have to employ adequate statistical tools to make reliable conclusions in the presence of DP noise. For example, in this work we will consider the 2022 American Community Survey (ACS) Census product covering the U.S. population and studying how several variables are related to homeownership. However, for the latter task there are no general and/or computationally feasible methodologies that can deliver reliable statistical inference on data of this scale unless, among others, strong parametric assumptions are made on the joint distribution of the data.

More in detail, while there are many DP methods to tackle a wide variety of statistical and machine learning problems, such as hypothesis testing \citep{awan2018differentially,kazan2023test}, empirical risk minimization \citep{Chaudhuri2011,bassily2014private} and the fitting of large machine learning models \citep{abadi2016deep}, there is still a need for general methods that can provide valid statistical analysis of privatized data for a broad amount of settings. Specifically, the challenge for statistical inference is that after the noise is introduced for privacy protection, the original data is now unobservable to the analyst, and the sampling distribution of the privatized data can be intractable \citep{williams2010probabilistic}. 

Among the existing inference methods for privatized data, Bayesian approaches are the most flexible and widely applicable. However, these methods often have significant limitations, such as at least one of the following: (i) requiring a fully generative model, (ii) relying on conjugate priors or (iii) suffering from a curse of dimensionality due to rejection sampling, thus limiting their applicability \citep{gong2019exact,ju2022data,chen2025particle}. In terms of frequentist inference, the main approaches are based on either asymptotic approximations \citep{wang2018statistical} or simulation-based inference such as the parametric bootstrap \citep{ferrando2022parametric,wang2025optimal,awan2025simulation}. Unfortunately, the most rigorous asymptotic approximations are only applicable to the simplest of problems \citep{wang2018statistical}; as for simulation-based methods, these also often rely on parametric models which can be too restrictive and also have substantial computational cost. Overall, these approaches are not able to allow for flexible models and do not scale favorably with the sample size and/or the dimension of the data.

\noindent{\bf Our Contributions: }We propose a large-sample Bayesian framework for statistical inference on privatized data that we refer to as Private Uncertainty Measurement via Bayesian Asymptotics (PUMBA). This method is designed for the case where a (multivariate) additive noise mechanism is applied to a set of statistics, which can be written as a sum over individuals; despite this limitation of scope we show that PUMBA produces state-of-the-art inferential results in many settings of interest.  PUMBA builds on the approach of \citet{guha2024causal} which, before our contributions, relied on strong assumptions about prior specification (potentially misaligned with analyst beliefs) and on asymptotic approximations that lack general theoretical justification. In a nutshell, this approach works by (i) imputing the confidential summary statistic and (ii) sampling from the non-private posterior given this summary; asymptotic approximations can be introduced in both steps so that these samplers no longer depend on the prior distribution or even on certain parts of the model. We show that this approach is more widely applicable than suggested in \citet{guha2024causal} and give a rigorous analysis of its asymptotic properties. The proofs of our results leverage theoretical tools from a variety of statistical disciplines such as local limit theorems, sequences of prior distributions, Bayesian asymptotics and information theory. While initially motivated by a Bayesian framework, we show through simulations that our methods achieve conservative frequentist properties (e.g., coverage/type I error) in the settings we consider. We apply our methodology to analyze drivers of homeownership using data from the 2022 American Community Survey, a U.S. Census Bureau data product. While this product was not protected with DP in 2022, we proposed a similar privatization scheme as was used in the 2020 DP Decennial Census products and show that we are able to achieve powerful and well-calibrated inference, whereas a naive approach leads to erroneous conclusions. 

\noindent{\bf Organization: } Section \ref{s:background} reviews the necessary background on differential privacy. In Section \ref{s:bayes_approx} we develop our main theoretical results: we describe the problem setup and introduce the approximations of interest (Section \ref{s:setup}); we develop the main approximation of the original summary statistic, given the privatized summary (Section \ref{s:TgivenS}) and then support the use of Bernstein von-Mises approximations (Section \ref{s:BvM}); finally we propose our approximate posterior sampling algorithm (Section \ref{s:sample}). We showcase our methodology through numerical experiments in Section \ref{s:numerical_exp}, applying our approximate Bayesian inference to both simulated and real data settings. We conclude in Section \ref{s:discussion} with some discussion. Proofs, technical details and ancillary results are included in the Supplementary Materials.

\noindent{\bf Related Work: }  The method discussed in this paper was initially proposed by \citet{guha2024causal} as a method for causal inference via the sub-sample and aggregate framework \citep{nissim2007smooth}. Building on this, our work aims to place itself within the body of research that develops general statistical inference tools for privatized data. We will therefore provide a brief overview of the related literature below. 

The Bayesian framework has been shown to be very flexible for analyzing privatized data. \citet{williams2010probabilistic} were the first to identify that the marginal likelihood, a key component in Bayesian analysis on privatized data, is computationally intractable. To circumvent the marginal likelihood, various data augmentation MCMC approaches have been proposed. When approximating the distribution of a sufficient statistic as a normal distribution (although not formally justified), \citet{bernstein2018differentially,bernstein2019differentially} show that a simple Gibbs sampler can approximate the posterior distribution. \citet{ju2022data} derived a Metropolis-within-Gibbs sampler that targets the exact private posterior distribution, but its performance is limited except when using conjugate priors and smaller sample sizes \citep{chen2025particle}. For simple problems, \citet{gong2019exact} proposed an i.i.d. sampler based on rejection sampling. This was expanded into a more general particle filter sampler by \citet{chen2025particle}, but is still limited to small sample sizes and relatively simple models. As an alternative approach, \citet{karwa2015private} derived a variational approximation for a naive Bayes classifier, but this requires specific derivations for each problem of interest.

There have also been frequentist approaches to statistical inference on privatized data. \citet{wang2018statistical} argue that naive asymptotic approximations do not have good finite sample performance on privatized data, and propose an alternative asymptotic regime, where only the sufficient statistics are approximated by asymptotics while the noise from the privacy mechanism is left unaltered. Another common approach is to use the parametric bootstrap to approximate the sampling distribution with privatized data \citep{ferrando2022parametric,alabi2022hypothesis}. However, \citet{wang2025optimal} show that the parametric bootstrap can often result in biased inference leading to poor coverage or type I errors and propose using indirect inference to first debias the DP statistic. As an alternative approach, \citet{awan2025simulation} propose a simulation-based inference strategy that guarantees conservative coverage and type I errors at the expense of additional computational cost.

While the above approaches address various challenges related to statistical inference for privatized outputs, as mentioned earlier they either require strong parametric assumptions (which become even more problematic as the dimension of the data scales) or become computationally intractable as the size (and/or dimension) of the data scales. These limitations motivate the method studied in this work.

%Among the prior work, there is a lack for widely-applicable and scalable solutions with strong theoretical guarantees. Our large-scale Bayesian method offers a unique approach that fills this gap.

\section{Background on Differential Privacy}\label{s:background}

% \subsection{Notation}
% For two probability distributions, $P$ and $Q$, define $\mathrm{TV}(P,Q)$ to be the total variation distance. We may also apply $\mathrm{TV}$ to random variables or density functions, which will represent the total variation distance between the underlying distributions. 

%\subsection{Differential Privacy}\label{s:dp}

We here provide a brief definition of DP in its original form (i.e. $\epsilon$-DP) which will be needed to frame the setting of our work, although our approach can be applied to other DP settings as well. More specifically, a randomized mechanism $\mathcal{M}$ is $\epsilon$-DP \citep{dwork2006differential} if, for all neighboring databases $D$ and $D'$ of size $n$ (i.e. databases that differ only by one entry), we have that
$$\mathbb{P}[\mathcal{M}(D) \in S] \leq \exp(\epsilon) \, \mathbb{P}[\mathcal{M}(D') \in S],$$
for all measurable sets $S$. This definition states that, once we apply a randomized mechanism $\mathcal{M}$ to a database, the probability of any event remains similar (within a factor of $\exp(\pm\epsilon)$) whenever an individual's data is arbitrarily changed in the database. 

A relaxation of $\epsilon$-DP is Gaussian DP (GDP) \citep{dong2022gaussian}. A mechanism $\mathcal{M}$ satisfies $\epsilon$-GDP if, for any test distinguishing $\mathcal{M}(D)$ from $\mathcal{M}(D')$ with type I error $\alpha$, the resulting type II error is at least as large as that of the optimal test between $N(0,1)$ and $N(\epsilon,1)$ at the same type I error $\alpha$. GDP has been growing in popularity and is being advocated as the current state-of-the-art DP framework \citep{gomez2025gaussian}.
% A relaxation of $\epsilon$-DP is Gaussian DP (GDP) \citep{dong2022gaussian}. A mechanism $M$ satisfies $\epsilon$-GDP if when the adversary attempts to test $H_0:M(D)$ versus  $H_1:M(D')$ at type I error $\alpha$, the type II error is at least as great as when testing $H_0: N(0,1)$ versus $H_1:N(\epsilon,1)$ at the same type I error. 

%There have been many other alternative definitions of DP, including the recently proposed Gaussian DP \citep{dong2022gaussian}, all of which either modify/relax the original definition or provide a different approach to interpreting privacy. 

%In all cases, the types of randomized mechanisms $m(\cdot)$ that allow to respect these different definitions of DP are relatively similar. 

A common privacy mechanism is an additive mechanism, defined as follows:
$$\mathcal{M}(D) = T_n(D) + \frac{\Delta_T}{\epsilon} Z,$$
where $T_n$ is the statistic (function) of interest (computed on a sample of size $n$), $Z$ is a random variable independent of $D$ and $\Delta_T=\sup_{D,D'} \lVert T_n(D)-T_n(D')\rVert$ is the \emph{sensitivity} of $T_n$ with respect to a norm $\lVert\cdot\rVert$.  The sensitivity $\Delta_T$ measures the maximum amount by which the statistic $T_n$ can change by arbitrarily changing one observation from the database: the higher the sensitivity, the larger the noise that needs to be added for privacy.  To satisfy $\epsilon$-DP, the $L_1$ norm is used and $Z \iid \mathrm{Laplace}(0,1)$; to satisfy $\epsilon$-GDP, the $L_2$ norm is used and $Z\iid N(0,1)$.  We will denote the output of such a DP mechanism as $\sdp=\mathcal{M}(D)$, and write $m(\sdp \, | \, T_n)$ for its conditional density (context will make it clear what base measure this is with respect to). The examples used in this paper will either satisfy $\epsilon$-DP or $\epsilon$-GDP, but our methodology can be applied to other DP frameworks as well. %Going forward, we simply write $T:=T_n$, where the dependence on $n$ is implicit.

%, so long as an additive noise mechanism is used.

\section{Large-Sample Bayesian Analysis on Privatized Data}
\label{s:bayes_approx}
In this section, we begin by setting up our problem of interest and the corresponding notation, and then develop the theory to support our proposed approximate sampler. 
\subsection{Bayesian Inference and Problem Setup}\label{s:setup}

Let us provide the general setting and notation for the considered methodology. Firstly, we will consider that the statistics of interest (that will be subject to a differentially private mechanism) are of the form:
$$T_n = \sum_{i = 1}^n t(x_i),$$
where $x_i$ are characteristics of the observational units (individuals) indexed by $i$ ($i = 1, \hdots, n$). For example, if considering the U.S. Decennial Census data, $n$ would represent the total population size in the U.S. whereas $t(x_i) \in \RR^k$ would represent a known vector function that encodes/quantifies the information on individual $i$: for example it can be a multivariate function that determines, among others, if an individual is a homeowner and in a specific county. Hence, following the latter example, an element of $T_n \in \RR^k$ would represent the sum of individuals who are homeowners in a specific county (which is indeed the type of statistic reported in national survey data). 

We model $t(x_1), \hdots, t(x_n)\,|\,\theta \iid p(\cdot \,|\,\theta)$, with $p$ being a probability density function (pdf) representing the conditional distribution for $t:=t(x)$ with respect to a base measure $\lambda$ on $\RR^k$ (commonly $\lambda$ is either the Lebesgue measure or the counting measure on $\ZZ^k$), and $\theta\in \Theta\subset \RR^d$ being the parameters which have a prior density $\pi$ with respect to base measure $\nu$ (commonly Lebesgue measure).  In a similar manner we let $p_n(\cdot \,|\,\theta)$ be the conditional distribution for $T_n$. The prior predictive distribution for $T_n$ is therefore
\[p_n(T_n)=\int_\theta p_n(T_n\,|\,\theta)\,\pi(\theta) \ d\nu(\theta).\]%Similarly, for a fixed $\theta_0 \in \Theta$, we let $p(t\,|\,\theta_0)$ be the frequentist pdf for $t$ with respect to $\lambda$ and use $p_n(T_n\,|\,\theta_0)$ to denote the distribution for $T_n$ where $t(x_i)\iid p(t\,|\,\theta_0)$. %In general, our goal is to obtain information on the parameter $\theta$ and, in particular, on the uncertainty on $\theta$ given the statistic $T_n$. Indeed, knowing the statistical properties of $T_n$ allows us to perform inference on the underlying parameter of interest $\theta$.

While a non-private analysis focuses on the posterior distribution $\pi(\theta\,|\,T_n)$, in our setting we do not observe $T_n$ but a privatized version $\sdp\,|\,T_n \sim m(\sdp\,|\,T_n)$ where $m$ is the density of the \emph{privacy mechanism} $\mathcal{M}$ (with respect to some base measure, but its specification is not needed for this paper). The full data generating process can be summarized by the following hierarchical model:

\begin{enumerate}[leftmargin = 2em]
    \item The parameters are sampled from a prior distribution, i.e.: $\theta \sim \pi(\theta)$;
    \item A statistic is sampled from the data generated through the parameter $\theta$, i.e.: $T_n\,|\,\theta \sim p_n(T_n\,|\,\theta)$;
    \item A privatized statistic is sampled through the privacy mechanism based on $T_n$, i.e.: $\sdp \, | \, T_n \sim m(\sdp\,|\,T_n)$.
\end{enumerate}

In the end, only $\sdp$ is published and an external analyst must infer about $\theta$ given $\sdp$. 
Thus, in private Bayesian inference, we are primarily interested in the \emph{private posterior distribution},
\[\pi(\theta\,|\,\sdp)\propto \pi(\theta) \int m(\sdp\,|\,T_n)\, p_n(T_n\,|\,\theta)\ d\nu(\theta),\]
which was first identified in \citet{williams2010probabilistic}. Instead of directly tackling $\pi(\theta\,|\,\sdp)$, we follow  \citet{gong2019exact,ju2022data,chen2025particle} and instead (approximately) sample from the joint posterior,
\[\pi(\theta,T_n\,|\,\sdp) \propto \pi(\theta) \,p_n(T_n\,|\,\theta)\,m(\sdp\,|\,T_n).\]
Based on this, isolating the values of $\theta$ finally delivers samples from $\pi(\theta\,|\,T_n)$. %It can be observed however that, to do so, we do not have all the required information to obtain this distribution. Therefore in the following sections we will provide the considered sampling approximations for this purpose and then show their theoretical validity. 

\subsection{Sampling Approximations}

The statistical challenge in the above-described setting is to perform reliable inference on the underlying parameter of interest $\theta$ when only the privatized statistic $\sdp$ is available to us. %Ideally, if we could directly access the statistic $T_n$, we could probably use standard statistical tools to perform inference on $\theta$, but this is not possible anymore in the DP setting since we cannot access $T_n$ due to the privacy constraint. 
In particular, we now have to deal with two sources of randomness: the first coming from the randomness over $T_n$ (described by $p_n$) while the second comes from the randomness over $\sdp$ (described by $m$). %The data-generating process for this setting can be described by the following Bayesian hierarchical model:
\begin{comment}
\begin{enumerate}[leftmargin = 2em]
    \item A parameter of interest is sampled from a prior distribution, i.e.: $\theta \sim \pi(\theta)$;
    \item A statistic is sampled from the data generated through the parameter $\theta$, i.e.: $T_n\,|\,\theta \sim p_n(T_n\,|\,\theta)$;
    \item A privatized statistic is sampled through the privacy mechanism based on $T$, i.e.: $\sdp \, | \, T_n \sim m(\sdp\,|\,T_n)$.
\end{enumerate}
Ideally, we would like to reverse this hierarchical model in order to summarize our uncertainty of $\theta$ given the observation $\sdp$. 
\end{comment}
This can be done through the following (exact) sampling algorithm:
\begin{enumerate}[start=1,label={\bfseries Step \arabic*:}, leftmargin = 5em]
    \item Sample $T_n \, | \, \sdp \sim p_n(T_n\,|\,\sdp) \propto m(\sdp\,|\, T_n)\,p_n(T_n)$;
    \item Sample $\theta\,|\,T_n \sim \pi_n(\theta\,|\,T_n)$.
\end{enumerate}
Here we have that $p_n(T_n\,|\,\sdp)$ and $\pi_n(\theta\,|\,T_n)$ are respectively the posterior distributions of (i) $T_n$ given $\sdp$ and (ii) $\theta$ given the non-privatized $T_n$. The subscript of $p_n$ and $\pi_n$ is there to remind us that these are the finite-sample distributions; later we will introduce asymptotic approximations, which will have the subscript ``$\infty$.'' We can verify that this sampler is exact as follows:
% \begin{comment}
% \begin{align}
%     \pi_n(\theta\,|\,T_n) \, \underbrace{m(\sdp\,|\,T_n) \, p_n(T_n)}_{\propto \, p_n(T_n\,|\,\sdp)} & \propto \pi(\theta) \, p_n(T_n\,|\,\theta) \, m(\sdp\,|\,T_n) \nonumber   
%     \propto \pi(\theta, T_n\,|\,\sdp).
%     \label{eq:post_approx}
% \end{align}
% \end{comment}
% \begin{align*}
%     \pi(\theta, T_n\,|\,\sdp)
%     \propto \pi(\theta) \, p_n(T_n\,|\,\theta) \, m(\sdp\,|\,T_n)
%     \propto \pi_n(\theta\,|\,T_n) \, m(\sdp\,|\,T_n) \, p_n(T_n)
%     \propto \pi_n(\theta|T_n) p_n(T_n|\sdp).
% \end{align*}
\begin{align*}
\pi(\theta, T_n\,|\,\sdp)
&\propto \pi(\theta) \, p_n(T_n\,|\,\theta) \, m(\sdp\,|\,T_n) \\
&\propto \pi_n(\theta\,|\,T_n) \, m(\sdp\,|\,T_n) \, p_n(T_n) \\
&\propto \pi_n(\theta \,|\, T_n) \, p_n(T_n \,|\, \sdp).
\end{align*}
 However, while Step 2 of the above algorithm consists in simply sampling from the non-private posterior distribution (for which there are often available techniques), in the first step, sampling $m(\sdp\,|\,T_n)\,p_n(T_n)$ is non-trivial, since the prior predictive distribution of $p_n(T_n)$ is often complex.

In this paper we validate the use of convenient approximations for both steps motivated by \citet{guha2024causal}. Indeed, \citet{guha2024causal} propose choosing a prior on $\theta$ that makes the prior predictive distribution $p_n(T_n)$ a constant, simplifying the distribution $p_n(T_n\,|\,\sdp) \propto m(\sdp\,|\,T_n)$, and using a normal approximation for $\pi_n(\theta\,|\,T_n)$, which we will denote by $\pi_\infty(\theta\,|\,T_n)$ and which implicitly depends on $n$. 
%In fact, as in \citet{guha2024causal}, $\pi_\infty(\theta\,|\,T_n)$ is often a Gaussian distribution whose mean is an efficient estimator for $\theta$, and whose covariance is $n$-times a quantity that converges to the Fisher information. From the latter it is also evident that the approximation $\pi_\infty(\theta\,|\,T_n)$ is often a function of $n$ and we will therefore hereinafter assume that this dependence on the sample size is implicit according to the context. 
While \citet{guha2024causal} show that these approximations empirically worked well in their setting, there still remain some unaddressed issues: (i) it is not always the case that there exists a prior distribution $\pi(\theta)$ that results in a flat prior predictive; (ii) even if such a prior exists, it may not align with an analyst's actual prior beliefs; and (iii) there is no rigorous justification for the use of the approximation $\pi_\infty(\theta\,|\,T_n)$ in this scheme. Following these points, we argue that (i) with large $n$, we asymptotically have that the total variation distance between $p_n(T_n\,|\,\sdp)$ and $\propto m(\sdp\,|\,T_n)$ goes to zero, under regularity assumptions, thus avoiding the need for a flat prior predictive distribution; and (ii) that the incorporation of the additional approximation $\pi_\infty(\theta\,|\,T_n)$ within this framework is  asymptotically justified as well, under mild conditions. 

\subsection{Approximation for $T_n\,|\,\sdp$ } \label{s:TgivenS}

We start by discussing the sampling distribution of $T_n\,|\,\sdp$ that is given by:
\begin{equation}
\label{eq:p_T_sdp}
    p_n(T_n\,|\,\sdp) = \frac{m(\sdp\,|\,T_n) \,p_n(T_n)}{\int m(\sdp\,|\,T_n)p_n(T_n)\,d\lambda(T_n)} \propto m(\sdp\,|\,T_n)\,p_n(T_n),
\end{equation}
which is a density, as a function of $T_n$, with respect to $\lambda$. As mentioned earlier, \citet{guha2024causal} simplify this distribution as $\propto m(\sdp\,|\,T_n)$  by assuming that a prior can be chosen to make $p_n(T_n)$ constant. %As mentioned earlier, in some problems (such as location-families) the property $p_n(T_n)=c$, for some constant $c$, can be achieved by using a flat prior for $\pi(\theta)$. %Another example is using a flat (proper) prior for the multinomial distribution. 
%In these cases,  the distribution of $T_n\,|\,\sdp$ is exactly proportional to $m(\sdp\,|\,T_n)$. 
However, it is not always possible to choose $\pi(\theta)$ in a way that makes $p_n(T_n=y)$ constant for all $y$ as demonstrated in the following example.

\begin{example}
    Suppose that $x_1,\ldots, x_n$ are i.i.d. with pmf $P(x_i~=0\,|\,\theta)=P(x_i~=~1\,|\,\theta)~=~\theta$ and $P(x_i=2\,|\,\theta)=1-2\theta$ for $\theta\in [0,1/2]$. Set $T_n=\sum_{i=1}^n x_i$. We will show that for any $n\geq 2$, and any (possibly improper) prior $\pi_n(\theta)$ with respect to a Lebesgue measure, there is no $c_n>0$ such that $P(T_n=y)=\int_0^{1/2}P(T_n=y\,|\,\theta)\,\pi_n(\theta)\ d\theta=c_n$ for all $y\in\{0,\ldots, 2n\}$. To this end, let $\pi_n(\theta)$ be given and suppose that $P(T_n=0)=P(T_n=1)$. Then, $P(T_n~=~0)~=~\int_0^{1/2}\theta^n \pi_n(\theta) \ d\theta$, whereas $P(T_n=1)=\int_0^{1/2}n\theta^n \pi_n(\theta)\ d\theta$. We see that $P(T_n=1)=nP(T_n=0)$, which together with our assumption that $P(T_n~=~0)~=~P(T_n~=~1)$ implies that $P(T_n=0)=nP(T_n=0)$. However, since $n\geq 2$, this implies that $P(T_n=0)=0$. Thus, there is no $c_n>0$ such that $P(T_n=y)=c_n$ for all $y$. 
\end{example}

Consequently, the justification given in \citet{guha2024causal} is not valid for all models. Furthermore, even if such a prior exists, it may not align with the analyst's prior beliefs. Therefore, instead of using an argument based on the existence of an appropriate prior, we justify the approximation in Equation \eqref{eq:p_T_sdp} through an asymptotic argument whose intuition is as follows: suppose that the variability in $m(\sdp\,|\,T_n)$, viewed as a distribution on $T_n$, is significantly ``narrower'' than that in  $p_n(T_n)$, which is common for additive noise mechanisms;  then, we propose that the approximation $m(\sdp\,|\,T_n)\,p_n(T_n) \approx m(\sdp\,|\,T_n)I(T_n\in \mathcal T_n)$ is justified, where $\mathcal T_n$ is a set containing the support of $T_n$ under $p_n$. This phenomenon is similar to how the likelihood dominates the prior distribution in non-private Bayesian analyses, and our analysis is motivated by results on sequences of prior distributions. This idea is illustrated in the following example, where the approximations indicated by ``$\approx$'' are not yet formally justified (the formal result is given later in Theorem \ref{thm:TgivenS}). %For this purpose, hereinafter we will use $I(\cdot)$ to represent the indicator function.

\begin{example}\label{ex:scratch}%\todo{JA: Note that we use $\ol x$ instead of sum.}
    Let $t_i\iid p$, where $p$ is a known distribution taking values on $[L,U]$ which allows us to define $\Delta=U-L$. Let $\mu=\EE t_i$ and $\sigma^2=\var(t_i)$ represent the expectation and variance of $t_i$ respectively. The Gaussian mechanism privatizes the sample mean $\ol t$ (our statistic of interest $T_n$) by sampling from $m(\sdp\,|\,\ol t)\sim N\left(\ol t,\nicefrac{\Delta^2}{n^2\ep^2}\right)$. Viewing $m(\sdp\,|\,\ol t)$ as a distribution on $\ol t$, we have that $m(\sdp\,|\,\ol t) \propto N_{[L,U]}\left(\sdp,\nicefrac{\Delta^2}{n^2\ep^2}\right)$, i.e. a truncated normal distribution. Using a large-sample approximation, we can approximate $p_n(\ol t) \approx \twid p(\ol t)\sim  N_{[L,U]}\left(\mu, \nicefrac{\sigma^2}{n}\right)$. With some algebra we obtain 
    \begin{align*}
	m(\sdp\,|\,\ol t)\,\twid p(\ol t)& \sim N_{[L,U]}\left( \left( \frac{n^2\ep^2}{\Delta^2}+\frac{n}{\sigma^2}\right)^{-1}\left(\frac{n^2\ep^2}{\Delta^2} \sdp + \frac{n}{\sigma^2} \mu\right)   , \left( \frac{n^2\ep^2}{\Delta^2}+\frac{n}{\sigma^2}\right)^{-1}\right)\\
				   &\approx N_{[L,U]}\left(\sdp,\frac{\Delta^2}{n^2\ep^2}\right) \propto m(\sdp\,|\,\ol t)\, I(\ol t\in [L,U]).
    \end{align*}
    %We see that in this example, sampling from $m(\sdp\,|\, \ol t)\,I(\ol t\in [L,U])$ (with respect to $\ol t$) amounts to sampling from a truncated normal distribution, which can be accomplished by rejection sampling.
\end{example}

% \begin{remark}
%     In some problems, such as location-families,  the property $f(x)=c$ can be achieved by using a flat prior for $\pi(\theta)$. Another example is using a flat (proper) prior for the multinomial distribution. In these cases,  the distribution of $x|\sdp$ is exactly $\propto m(\sdp|x)$. However, it is not clear if it is always possible to choose $\pi(\theta)$ in a way that makes $f(x)$ constant for all $x$. 
% \end{remark}

We now set the ground to demonstrate why the above approximation is theoretically justified in common situations. In contrast to Example \ref{ex:scratch}, where the statistic of interest is the sample mean $\ol t$ and the noise mechanism depends on $n$, hereinafter we will take our statistic $T_n$ to be a sum of values and require that the mechanism does not depend on $n$. This amounts to a simple rescaling by $n$, but is beneficial to make the theoretical results precise. That being said, our proposed large-sample approximation of $p_n(T_n\,|\,\sdp)$ is
\begin{equation}\label{eq:p_infty}p_\infty^{\mathcal T_n}(T_n\,|\,\sdp) = \frac{m(\sdp\,|\,T_n)}{\int m(\sdp\,|\,T_n)\,d\lambda(T_n)} \propto m(\sdp\,|\,T_n)I(T_n\in \mathcal T_n).\end{equation}
For this approximation to be valid, we firstly need to make some mild but necessary assumptions to ensure that the approximating density $m(\sdp\,|\,T_n=t)$, as a function of $t$, is a well defined density with certain useful properties. Assumptions which are related to the privacy mechanism will be prefixed by ``M'' (for mechanism), while those related to the model $p_n$ or prior $\pi$ will be prefixed by ``D'' (for distribution).
\begin{assumpM}\label{a.location_mechanism}
The density $m(\sdp\,|\,T_n)$ is a location-family, does not depend on the sample size $n$, and is bounded by a constant $L<\infty$.
\end{assumpM}
Assumption \ref{a.location_mechanism} states that the mechanism $\mathcal{M}$ only depends on the quantity $(\sdp-T_n)$, which is satisfied by additive noise mechanisms, such as the Laplace mechanism \citep{dwork2006calibrating} or the Gaussian mechanism \citep{dong2022gaussian}. Our proof of Theorem \ref{thm:TgivenS} uses Assumption \ref{a.location_mechanism} to recenter both $T_n$ and $\sdp$ such that  $\EE T_n=0$ before applying the Gaussian approximation, and then translate back afterwards. Without this translation, the Gaussian approximation would not converge to a constant function, but instead to a function of the form $g(T_n)\propto \exp(cT_n)$, which is not the desired form for our approximation. It is possible that the proof can be adapted to handle this non-constant $g$, but this is beyond this paper's scope. In addition, the requirement that the mechanism does not depend on $n$ restricts the type of statistics $T_n$ to be sums, such as counts, or averages that can be rescaled such that they no longer depend on $n$. %Another required assumption for our first result is the following.

\begin{assumpM}\label{a.support_mechanism}
The sets $\mathcal{T}_n\subset \RR^k$ are nested and increasing, and each contains the support of $T_n$ under $p_n(\cdot \, |\,\theta)$ for all $\theta\in \Theta$. Furthermore, $m(\sdp\,|\, T_n=y)$ and  $m(\sdp\,|\,T_n~=~y)\,I(y~\in~\mathcal{T}_n)$ are integrable over $y$ with respect to $\lambda$, $\int m(\sdp\,|\,T_n=y)\ d\lambda(y)<\infty$, and there exists a constant $c>0$ such that $\int m(\sdp\,|\,T_n=y)\,I(y \in \mathcal{T}_1)\, d\lambda(y) \geq c$ for all $\sdp$.
\end{assumpM}
%
%\begin{assumpM}\label{a.integrable_mechanism}
%The density $m(\sdp\,|\,T_n)$ is integrable over $T_n$ with respect to $\lambda$, and such that $\int m(\sdp\,|\,T_n=y)\,d\lambda(y)>0$.
%\end{assumpM}
%
Assumption \ref{a.support_mechanism} requires the mechanism density $m$ to be an integrable and positive function over $T_n$, which is always fulfilled by common additive privacy mechanisms that respect Assumption \ref{a.location_mechanism}. It also introduces the nested sets $\mathcal T_n$ which contain the support of $T_n$; the asymptotic results are valid even if these sets are conservative.

%\begin{remark}
%    Instead of restricting the support of the approximating distribution $p_\infty^{\mathcal{T}_n}(T_n\,|\,\sdp)$, we could instead apply some type of projection operator to force values outside of $\mathcal{T}_n$ to return within this support, while leaving those in $\mathcal{T}_n$ fixed. For example, Assumption \ref{a.support_mechanism} could be adapted based on this approach by instead requiring the use of a map $P_{\mathcal{T}_n}:\RR^k\rightarrow \RR^k$ such that for all $y \in \mathcal{T}_n$ we have that $P_{\mathcal{T}_n}(y)=y$, consequently obtaining a result analogous to Theorem \ref{thm:TgivenS} (further on) by using the data processing inequality. While we expect that $p_\infty^{\mathcal{T}_n}(T_n\,|\,\sdp)$ is generally the superior approximation, as it more closely coincides with the distribution arising from a principled Bayesian analysis (see Example \ref{ex:scratch}), it may be worth exploring the ``projection'' approximation as well.
%\end{remark}

We now move on to make a few assumptions on the distribution $p(t|\theta)$ and the prior $\pi(\theta)$. The first of these assumptions considers some properties of the  distribution $p(t \,|\, \theta)$ according to whether it is defined on a continuous or discrete support.

%Having placed some assumptions related to the properties of the privacy mechanism $m$, our goal is now to show that the above approximations are indeed valid for the density $p_n(T_n\,|\, \sdp)$. To deliver our first result in this sense, w
%, which are used to apply \emph{local limit theorems}, which are Gaussian approximations to the density.

\begin{assumpD}\label{a.bounded_density}
The distribution $p(t\,|\,\theta)$ and base measure $\lambda$ satisfies either of the following conditions:
\begin{enumerate}
    \item[(a)] $\lambda$ is the Lebesgue measure on $\RR^k$ and for each $\theta\in \Theta$, $p(t\,|\,\theta)$ is a bounded density (the bound may possibly depend on $\theta$);
    \item[(b)] $t\sim p(t\,|\,\theta)$ takes values in $\ZZ^k$ with $e_1,\ldots, e_k$ forming a basis for a minimal lattice containing the support of $t$ which satisfies $\det(e_1,\ldots, e_k)=1$ (the basis does not depend on $\theta$), and $\lambda$ is the counting measure on this lattice.  
\end{enumerate} 
\end{assumpD}

The requirement of a bounded density is easily understood, however in the discrete case the conditions on the support may be harder to interpret. As an example, if $t$ can take on the values $(1,0,\ldots, 0)^\top$, $(0,1,0,\ldots, 0)^\top$, $\ldots$, $(0,\ldots,0,1)^\top$, then these values form a basis $e_1,\ldots, e_k$ for $\ZZ^k$, which is a minimal lattice (no elements can be removed without altering the span); furthermore, $\det(e_1,\ldots, e_k)=1$. As a counterexample, if $t$ only takes values on even integers, including $e_1=(2,0,\ldots, 0)^\top$, $\ldots$, $e_k=(0,\ldots, 0,2)^\top$, then we have $\det(e_1,\ldots, e_k)=2^k$, which does not satisfy the assumption; in this case, one should first rescale the function as $t'=t/2$. A final assumption needed for our first result in Theorem \ref{thm:TgivenS} considers the moments of the distribution $p(t\,|\,\theta)$.

\begin{assumpD}\label{a.bounded_moments}
For all $\theta\in \Theta$, $p(t\,|\,\theta)$ has a finite mean and a positive definite covariance with finite entries. %$\EE_{t\sim p(t|\theta)} t$ is finite and $\Sigma = \EE_{t\sim p(t|\theta)}(t-\EE_pt)(t-\EE_pt)^\top$ is positive definite with finite entries.
\end{assumpD}

 Together, Assumptions \ref{a.bounded_density} and \ref{a.bounded_moments} allow us to employ \emph{local limit theorems} to justify a Gaussian approximation for the density of $T_n$ \citep{shervashidze2006convergence,gamkrelidze2015local,meizler1949multidimensional}, similar to Example \ref{ex:scratch}. We need a \emph{local limit theorem} in this case, rather than merely a \emph{central limit theorem}, because the private posterior distribution requires manipulations of the density rather than the distribution function.

\begin{remark}
Note that for a statistic $T_n=\sum_{i=1}^n t_i$ to have finite sensitivity, it is common for $t_i$ to take values in a bounded set. %often a necessary condition for an additive privacy mechanism is for $t$ to take values in a bounded set. As a result, 
If this is the case, this automatically implies that all moments are finite thereby reducing Assumption \ref{a.bounded_moments} to only requiring that the covariance of $t$ is invertible. 
\end{remark}

The final assumption for Theorem \ref{thm:TgivenS} is the only one on the prior distribution, merely requiring that it is a proper prior with respect to the base measure $\nu$, meaning that it integrates to one.

\begin{assumpD}\label{p.prior}
    The prior $\pi(\theta)$ is proper with respect to some base measure $\nu(\theta)$.
\end{assumpD}

Having listed our required assumptions, we are now ready to state the main result of this section, which establishes the asymptotic validity of our proposed approximation for $p_n(T_n\,|\,\sdp)$. However, it is convenient to  first show that the approximation works in the ``frequentist setting'', to approximate $p_n(T_n\,|\,\sdp,\theta)$, and then extend it to the Bayesian setting. With this in mind, let us denote the total variation distance between two distributions $P$ and $Q$ as $\mathrm{TV}(P,Q) = \sup_A |P(A)-Q(A)|$, where the supremum is over all measurable sets $A$. We will commonly write $\mathrm{TV}(f,g)$ in place of $\mathrm{TV}(P,Q)$ when $f$ and $g$ are densities for $P$ and $Q$ with respect to a common base measure. For all results, the detailed proofs and auxiliary results can be found in Appendix \ref{app:proofs}, while here we limit ourselves to brief proof sketches after each result.
 % under both $p_n(T_n|\sdp,\theta)$ and $p_n(T_n|\sdp)$.

\setcounter{thm}{0}

\begin{restatable}{thm}{thmTgivenS}
    \label{thm:TgivenS}%\todo{JA: simplify as just TV?}
    Suppose that Assumptions \ref{a.location_mechanism}-\ref{a.support_mechanism} and \ref{a.bounded_density}-\ref{a.bounded_moments} hold. Let $\sdp^{(n)}$ be any sequence of random variables (such as from $\sdp^{(n)}\,|\,T_n$ and $T_n\sim p_n(T_n\,|\,\theta)$). %Then, 
    %Under Assumptions \ref{a.location_mechanism}-\ref{a.support_mechanism} and \ref{a.bounded_density}-\ref{a.bounded_moments}, we have that 
    
    \begin{enumerate}
    \item Then, for every $\theta\in \Theta$,
    %\[\sup_{\sdp} \left| p_n(T_n\,|\,\sdp,\theta)\,-\,p_\infty^{\mathcal T_n}(T_n\,|\,\sdp)\right|\rightarrow 0.\]
    %It follows that if $\sdp^{(n)}$ is any sequence of random variables (such as from $\sdp^{(n)}|T_n$ and $T_n\sim p_n(T_n|\theta)$), then for every $\theta\in \Theta$,
    $\mathrm{TV}\big(p_n(T_n\,|\,\sdp^{(n)},\theta)\,,\,p_\infty^{\mathcal T_n}(T_n\,|\,\sdp^{(n)})\big) \overset {a.s.} \rightarrow 0.$
    
    \item If Assumption \ref{p.prior} also holds, then
    $\mathrm{TV}\big(p_n(T_n\,|\,\sdp^{(n)})\,,\,p_\infty^{\mathcal T_n}(T_n\,|\,\sdp^{(n)})\big) \overset {a.s.} \rightarrow 0.$
    %where $\sdp^{(n)}$ is any sequence of random variables (such as from $\sdp^{(n)}|T_n$ and $T_n\sim p_n$).
    \end{enumerate}
\end{restatable}

\begin{proof}[Proof Sketch.]
The frequentist result is established first: a local limit theorem is used to approximate the pmf/pdf of $T_n$ as a Gaussian. We analyze these Gaussians in a similar way as a sequence of prior distributions, showing that their contribution converges to a constant function. From this, we establish that the densities $p_n(T_n\,|\,\sdp^{(n)},\theta)$ and $p_\infty^{\mathcal T_n}(T_n\,|\,\sdp^{(n)})$ converge uniformly over $\sdp$, establishing the convergence in total variation. The result for the prior predictive follows by taking the expectation over $\theta\sim \pi$ of the previous total variation result, applying the Dominated Convergence Theorem. %See Lemmas  \ref{lem:lltcont}, and \ref{lem:lltint} in Appendix \ref{app:proofs} for details. 
\end{proof}
%\end{thm}

Theorem \ref{thm:TgivenS} shows that, under reasonable assumptions, $p_\infty^{\mathcal T_n}(T_n\,|\,\sdp)$ is a valid approximation for $p_n(T_n\,|\,\sdp)$. Recall that our final goal is to obtain the approximation for the private posterior distribution, namely $\pi_n(\theta\,|\,T_n)\, p_n(T_n\,|\,\sdp) \propto \pi(\theta, T_n\,|\,\sdp)$. Using the previous result along with the data processing inequality, we can consequently state the following:

\setcounter{cor}{0}

\begin{cor}\label{cor:tv_first_approx}
Under the same assumptions as Theorem \ref{thm:TgivenS}, 
\begin{equation}    \label{eq.tv_first_approx}
    \mathrm{TV}\big(\pi_n(\theta\,|\,T_n)\,p_n(T_n\,|\,\sdp^{(n)})\,,\,\pi_n(\theta\,|\,T_n)\,p_\infty^{\mathcal T_n}(T_n\,|\,\sdp^{(n)})\big) \overset {a.s.}\rightarrow 0.
\end{equation}
\end{cor}
\begin{proof}
    The result follows by applying the data processing inequality for total variation to the result of Theorem \ref{thm:TgivenS}.
\end{proof}
It can be noticed that the left side of Equation \eqref{eq.tv_first_approx} is the posterior distribution of $(\theta,T_n)\,|\,\sdp$ and the right side uses our approximation of $T_n\,|\,\sdp$, but then still uses the ``non-private posterior'' $\pi_n(\theta\,|\,T_n)$. This justifies the replacement of $p_n(T_n\,|\,\sdp)$ with $p_\infty^{\mathcal T_n}(T_n\,|\,\sdp)$ in our sampling algorithm discussed in Section \ref{s:setup}. In the next section, we offer theoretical justifications for using an approximation of $\pi_n(\theta\,|\,T_n)$ as well.%, such as by a Laplace approximation founded on the Bernstein von-Mises Theorem. 

\subsection{Approximation for $\theta\,|\,T_n$ }\label{s:BvM}

If we have access to an exact sampler for $\pi_n(\theta\,|\,T_n)$ (such as with conjugate priors), then the results of Theorem \ref{thm:TgivenS} and Corollary \ref{cor:tv_first_approx} may be sufficient. However, since we are already approximating $p_n(T_n\,|\,\sdp)$ we may be interested in approximating $\pi_n(\theta\,|\,T_n)$ as well. In this section, we show that if  $\pi_\infty(\theta\,|\,T_n)$ is a standard large sample approximation for $\pi_n(\theta\,|\,T_n)$, then $\pi_\infty(\theta\,|\,T_n)\,p_\infty^{\mathcal T_n}(T_n\,|\,\sdp)$ is a close approximation to $\pi_n(\theta,T_n|\sdp)$. For example, recall that by the Bernstein von-Mises Theorem \citep[Section 10.2]{VanDerVaart2000}, the posterior distribution for many models is approximately normal:
\[\pi_n(\theta\,|\,T_n) \approx N\left(\hat\theta(T_n), n^{-1}\mathcal{I}^{-1}(\hat\theta(T_n))\right),\]
where $\hat\theta(T_n)$ is any efficient estimator and $\mathcal{I}$ is the Fisher information for a single observation. Given this, we need one final assumption to justify the use of $\pi_\infty(\theta\,|\,T_n)$ in our approximation, which is similar to the Bernstein von-Mises Theorem  \citep[Section 10.2]{VanDerVaart2000}.

%$\pi_n(\theta\,|\,T_n)\,p_n(T_n\,|\,\sdp)$.  %For an approximation that incorporates the prior, we could use Laplace's approximation: 
%\[\pi_n(\theta\,|\,T_n) \approx N(\hat\theta,\Phi),\]
%where $\hat\theta=\arg\max_\theta \log\pi_n(\theta\,|\,T_n)$ is the maximum aposteriori estimator and the inverse covariance matrix is $\Phi^{-1}=-\nabla^2_\theta \log \pi_n(\theta\,|\,T_n)|_{\theta=\hat\theta}$, which is the negative Hessian of $\log\pi_n(\theta\,|\,T_n)$ evaluated at $\hat\theta$. 

%We justify this based on a Bernstein von-Mises-style approximation, which can be used in combination with the approximations for $p_n(T_n\,|\,\sdp)$ developed in the previous section. For this reason we give a final assumption which, being unrelated to conditions on $p$ and $m$, belongs to a standalone class.

\begin{assumpD}\label{a.bvm}
The density $\pi_\infty(\theta\,|\,T_n)$ satisfies $\mathrm{TV}\big(\pi_n(\theta\,|\,T_n),\pi_\infty(\theta\,|\,T_n)\big)\overset p\rightarrow 0$, with the probability being with respect to the frequentist distribution $p_n(T_n\,|\,\theta_0)$ at a fixed $\theta_0$. 
\end{assumpD}

%Many standard approximations can satisfy Assumption \ref{a.bvm}, such as the Bernstein von-Mises Theorem or Laplace approximations \citep[Section 10.2]{VanDerVaart2000}. %Indeed Assumption \ref{a.bvm}  allows us to replace the true posterior distribution of $\theta\,|\,T_n$ with a simpler approximation (usually a normal approximation as described earlier) which, aside from making sampling simpler, can also simplify the modeling and prior choices as the approximation may have a simpler form with fewer parameters. In this work, we assume that the non-private Bernstein von-Mises approximation satisfies Assumption \ref{a.bvm}. 

%In addition, for the next results we establish a modified notation for the assumptions regarding the prior predictive distribution $p$: the notation D$\cdot$' represents Assumption D$\cdot$ but adds the same requirements for the frequentist case where $p$ is replaced with the frequentist distribution $p(t\,|\,\theta_0)$. Explicit statements of these assumptions are found in the Appendix as ??? \todo{add assumptions to appendix.}
%\begin{thm}\label{thm:fullTV}

We can now deliver the final theoretical result that justifies the proposed approach. Note that that the distribution of $\sdp\,|\,\theta$ depends on $n$ (since $T_n$ depends on $n$).

\begin{restatable}{thm}{thmFullTV}\label{thm:fullTV}
    Under Assumptions  \ref{a.location_mechanism}-\ref{a.support_mechanism} and \ref{a.bounded_density}-\ref{a.bvm}, we have that 
    \[\mathrm{TV}\big(\pi_n(\theta\,|\,T_n)\,p_n(T_n\,|\,\sdp),\pi_\infty(\theta\,|\,T_n)\,p_\infty^{\mathcal T_n}(T_n\,|\,\sdp)\big)\overset p \rightarrow 0,\]
    where the probability is over the marginal distribution of $\sdp\,|\,\theta_0$ as $n\rightarrow \infty$. 
\end{restatable}

%Finally, similarly to the adaptation made with Corollary \ref{cor:TgivenS}, we deliver the same result when constraining the support of $T_n$ to $\mathcal{T}_n$.

\begin{comment}
\begin{restatable}{corollary}{corFullTV}\label{cor:fullTV}
    Under Assumptions \ref{a.bounded_density}'-\ref{a.bounded_moments}', \ref{a.location_mechanism},\ref{a.support_mechanism} and \ref{a.bvm}, we have that 
    \[\mathrm{TV}\big(\pi_n(\theta\,|\,T_n)\,p_n(T_n\,|\,\sdp),\pi_\infty(\theta\,|\,T_n)\,p_\infty^{\mathcal{T}_n}(T_n\,|\,\sdp)\big)\overset p \rightarrow 0,\]
    where the probability is over the marginal distribution of $T_n\,|\,\theta_0$ as $n\rightarrow \infty$.
\end{restatable}
\end{comment}

%Similarly to Remark \ref{rem:M4}, also in this case we could replace $p_\infty^{\mathcal{T}_n}(T_n\,|\,\sdp)$ with the ``projection'' ersion of the approximating distribution, this time under the frequentist distribution $p_n(T_n\,|\,\theta_0)$. Below is a brief proof sketch for these two last results.

\begin{proof}[Proof Sketch.]
We upper bound the total variation through four terms by the triangle inequality, such that each term only differs in one of the two densities. Three terms simplify  to either $\mathrm{TV}(p_n(T_n\,|\,\sdp),p_\infty^{\mathcal T_n}(T_n\,|\,\sdp))$ or  $\mathrm{TV}(p_n(T_n\,|\,\sdp,\theta_0),p_\infty^{\mathcal T_n}(T_n\,|\,\sdp))$ by the data processing inequality, which converge to zero almost surely by Theorem \ref{thm:TgivenS}. The remaining term can be expressed as $\EE\limits_{T_n\sim p_n(T_n\,|\,\sdp,\theta_0)} \mathrm{TV}\big(\pi_n(\theta\,|\,T_n),\pi_\infty(\theta\,|\,T_n)\big)$. Since the term inside the expectation converges to zero in probability with respect to $T_n \,|\, \theta_0$, by Lemma \ref{lem:condExp} in Appendix \ref{app:proofs}, the conditional expectation also converges to zero in probability.% (with respect to $\sdp|\theta_0$ as $n\rightarrow \infty$).
\end{proof}

Besides simplifying the sampling process, the approximation  in Theorem \ref{thm:fullTV} also simplifies the modeling step: as long as the prior satisfies Assumptions \ref{p.prior} and \ref{a.bvm}, there is no need to specify or incorporate the prior into our procedure. Thus, while the justification of \citet{guha2024causal} is \emph{only} valid for priors that result in flat prior predictive distributions, our analysis is valid for \emph{nearly all} prior distributions, only requiring a proper prior and a Bernstein von-Mises approximation.

\subsection{Sampling Algorithm}\label{s:sample}

Following the above results, we present the proposed sampling algorithm to obtain an approximate posterior sample from $\theta\,|\,\sdp$. More specifically, as per Assumption \ref{a.location_mechanism}, we assume that we receive a private output $\sdp$ that is a result of an additive privacy mechanism $\mathcal{M}$ (e.g. Laplace or Gaussian) which is based on a summary statistic $T_n$. The proposed sampling strategy is presented in Algorithm \ref{algo.sample} in which the statistic $T_n$ is sampled from the distribution $p_n$ (or $p_{\infty}$) and then we  sample $\theta_i\,|\,T^{(i)}_n$. By our above theory, the generated samples approximately follow the posterior distribution $\theta\,|\,\sdp$. More specifically, aside from $\sdp$ and $m$, for Algorithm \ref{algo.sample} we need to specify the number of samples to produce as well as our choice of distributions for $p_\infty^{\mathcal T_n}(T_n\,|\,\sdp)$ and $\pi^*(\theta\,|\, T_n)$ which, in the latter case, can either be $\pi_n(\theta\,|\, T_n)$ or the approximation $\pi_{\infty}(\theta\,|\, T_n)$. Note that the approximation for sampling the statistics $T^{(i)}_n$ is $p_{\infty}^{\mathcal{T}_n}(T_n\,|\,\sdp) \propto m(\sdp\,|\,T_n) \, I(T_n \in \mathcal{T}_n)$. Moreover, since the $\hat{\theta}_i$ are i.i.d., we can easily account for Monte Carlo errors.

%However, for versions of our proposed algorithm, we need to make an additional assumption.

%\begin{assumpA}\label{a.theta_hat}
%There exists an efficient estimator $\hat\theta$ for $\theta$ which only depends on $T_n$
%\end{assumpA}

%For the settings we consider, this is not considered a strong assumption. For example, this assumption holds in the case where the data comes from an exponential family distribution. Based on this, t

\vspace{1em}

\begin{algorithm}
\footnotesize % Smaller font to reduce vertical space
\caption{Private Uncertainty Measurement via Bayesian Asymptotics (PUMBA)}
\begin{algorithmic}[1]
\setlength{\itemsep}{1pt} % Reduce space between lines
\REQUIRE Private output $\sdp$; Privacy mechanism $M$; Distributions $p_{\infty}(T_n\,|\,\sdp)$ and $\pi^*(\theta\,|\,T_n)$; Number of Monte Carlo samples $R$; Support $\mathcal{T}_n$; Sample size $n$.
\ENSURE Approximate posterior samples $\theta_1, \dots, \theta_R$.
\STATE Sample $T_n^{(1)},\ldots, T_n^{(R)} \overset{\text{iid}}{\sim} p_\infty^{\mathcal T_n}(T_n\,|\,\sdp)$ 
\FOR{each $T_n^{(i)}$}
    \STATE $\theta_i \sim \pi^*(\theta\,|\,T_n^{(i)})$.
%    \STATE \textit{Optional}: Compute additional statistics $J(\theta_i)$ 
\ENDFOR
\end{algorithmic}
\label{algo.sample}
\end{algorithm}

\begin{remark}
\label{rmk:alt_algo}
It is possible to consider alternative sampling strategies according to the analysis requirements. For example, in the case that the private posterior distribution is approximately normal and we have access to an efficient estimator $\hat\theta(T_n)$ for $\theta$, instead of sampling $\theta_i\sim \pi^*(\theta\,|\,T_n^{(i)})$ we can directly estimate the posterior mean and covariance which can then be used to produce credible regions. This alternative procedure is presented in Algorithm \ref{algo.totcov}, found in Appendix \ref{s.totcov_approx}.
\end{remark}

\begin{remark}
    Implementation of step 1 in Algorithm \ref{algo.sample} should be straightforward. When using additive distributions such as Laplace or Gaussian noise, then $p_\infty^{\mathcal T_n}$ amounts to sampling a truncated Laplace/Gaussian distribution, which can be accomplished by rejection sampling, or  approximate samplers such as MCMC, slice samplers or particle filters. 
\end{remark}

\section{Simulation Studies}
\label{s:numerical_exp}
In this section, we compare the performance of PUMBA in two controlled simulations studies: one for the mean of bounded data and the other for inference on the slope in a simple linear regression model (an additional more complex emulation scenario can be found in Section \ref{s.home_sim}). In both cases, we highlight that PUMBA has significantly more lenient modeling assumptions compared to previous solutions. 

\subsection{Confidence Interval for the Mean of Bounded Data}\label{s:mean}

In this section, we tackle a relatively simple problem that, to date, has not received a satisfactory private solution. More specifically, given $t_1,\ldots, t_n$, which are i.i.d.  random variables lying in $[0,1]$, how do we produce a DP confidence interval for the mean? Call $\mu = \EE t_i$ and $\sigma^2 = \var(t_i)$. The statistics we use are noisy first and second moments, which is a similar approach to that used in \citet{du2020differentially}, \citet{wang2025differentially} and \citet{awan2025simulation} where, however, they use a noisy sample variance in place of a noisy second moment. Nevertheless, these approaches have some limitations with respect to the postulated scenario: (i) \citet{karwa2018finite,du2020differentially} and \citet{awan2025simulation} assume the normal parametric family which is not valid in our general scenario, while (ii) \citet{wang2025differentially} can handle our non-Gaussian setting, but they only deliver an \textit{asymptotic GDP guarantee} and require very large sample sizes (their simulations start at $n=10,000$), making it inapplicable to the simulation setting considered here. 

More in detail, we choose our vector of DP statistics $\sdp$ to be composed by $s_1 = \ol t+\mathrm{Laplace}(0,\nicefrac{1}{(n\,\ep_1)})$ and $s_2 = \frac1n \sum_{i=1}^n t_i^2+\mathrm{Laplace}(0,\nicefrac{1}{(n\,\ep_2)})$, which jointly satisfy $(\ep_1+\ep_2)$-DP. Also, we have that $\theta = [\mu, \sigma^2]$ and $T = [\,\ol t, \frac 1n \sum_{i=1}^n t_i^2]$.  We implement our approximate Bayesian method using Algorithm \ref{algo.sample} with a rejection sampler, using the Gaussian approximation $\pi_n(\theta\,|\,T_n) \sim N(\ol t,s_t^2/n)$, where $s_t^2$ is the sample variance (which can be computed from $T$) and the the constraint set $\mathcal{T}_n$ is as follows:
\begin{equation}
\label{eq:moments}
0\leq \left(\frac 1n \sum_{i=1}^n t_i\right)^2\leq \frac 1n \sum_{i=1}^n t_i^2\leq \frac 1n \sum_{i=1}^n t_i \leq 1,
\end{equation}
where the second inequality is by Jensen and the other inequalities use the fact that $t_i \in [0,1]$. We compare against an oracle benchmark derived in \citet{wang2018statistical}, which uses a frequentist approximating distribution for $s_1$, assuming that the variance of $t_i$ is known (the framework of \citet{wang2018statistical} cannot handle the case of unknown variance), as well as a plug-in version  of \citet{wang2018statistical} that replaces the population variance with a DP estimator.

We run a simulation to study how these different methods perform in terms of coverage and width of their delivered confidence intervals (or credible intervals for our approach). We generate samples $t_i$ ($i = 1, \hdots, n$) of size $n = 1,000$ from the distribution $\mathrm{Beta}(2,4)$, which has mean $\mu = \nicefrac{1}{3}$ and variance $\sigma^2 = \nicefrac{2}{63} \approx .0317$. Recalling that we use the additive Laplace mechanism with $\epsilon_1 = 0.95$ and $\epsilon_2 = 0.05$ for the first and second estimated moments respectively, for \citet{wang2018statistical} the sampling distribution of $s_1$ is approximated as the convolution $N(\mu,\nicefrac{\sigma^2}{n})+\mathrm{Laplace}(0,\nicefrac{1}{(n\ep)})$, while the ``Plug-In'' method uses the same approximation as \citet{wang2018statistical} but replaces the true value $\sigma^2$ with the plug-in estimator $\hat{\sigma}^2 = \max\{s_2-s_1^2,0\}$. For our approach, we set the number of Monte Carlo samples to $R=1,000$ and use the approximations $p_\infty^{\mathcal{T}_n}(T_n\,|\,\sdp)$ and $\pi_\infty(\theta\,|\,T_n)$. The results of this simulation can be found in Table \ref{tab:mean}.

% \begin{table}[t]
%     \centering
%     \renewcommand{\arraystretch}{1.2}
%     \setlength{\tabcolsep}{10pt}
%     \begin{tabular}{lcc}
%         \toprule
%         \textbf{Method} & \textbf{Coverage} & \textbf{Width} \\
%         \midrule
%         \citet{wang2018statistical} (Known Variance) & 0.946 (0.007) & 0.0228 ($<0.0001$) \\
%         Plug-In & 0.867 (0.011) & 0.0216 (0.0003) \\
%         PUMBA & 0.969 (0.005) & 0.0260 (0.0002) \\
%         \bottomrule
%     \end{tabular}
%     \caption{Coverage and width of DP credible/confidence intervals for the mean of $[0,1]$ random variables, with Monte Carlo standard errors in parentheses. Simulation settings: $n=1,000$, $\epsilon_1=0.95$, $\epsilon_2=0.05$, $R=1,000$ (number of Monte Carlo samples), target coverage 0.95, results aggregated over 1,000 replicates.}
%     \label{tab:mean}
% \end{table}

It can be seen how, in this experiment, the oracle method of \citet{wang2018statistical} achieves accurate coverage and has a baseline width of $0.023$. On the other hand, replacing the true variance with the plug-in estimator results in significantly lower coverage of $0.866$, along with slightly smaller width. PUMBA achieves a good coverage of $0.952$ and a slightly larger width of $0.026$, providing well-calibrated statistical inference without requiring the true variance to achieve this. 

%In Section \ref{s:mean2} of the Appendix, we consider a slightly different scenario where the sample mean and sample variance are privatized. While our method still gives good performance, this setting does not technically align with our theoretical assumptions since the sample variance cannot be expressed as a sum of independent random variables. The real-data example in the next section also provides some additional simulation results (based on the true data) that highlight the validity of our proposed approach. \todo{Do we have the simulation with mean and variance privatized?}

% \begin{table}[t]
%     \centering
%     \begin{tabular}{l|ll}
%         Method & Coverage & Width \\\hline
%         \citet{wang2018statistical} (Known Variance)&0.946 (0.007) & 0.0228 ($<0.0001$)\\
%         Plug In & 0.867 (0.011) & 0.0216 ($0.0003$)\\
%         Ours (Rejection Sampling) & 0.969 (0.005) & 0.0260 ($0.0002$)
%     \end{tabular}
%     \caption{Coverage and width of DP credible/confidence intervals for the mean of $[0,1]$ random variables, with Monte Carlo standard errors in parentheses. See Section \ref{s:mean} for setting details. Simulation settings are $n=1000$, $\epsilon_1=.95$, $\epsilon_2=.05$, $R=1000$ (number of Monte Carlo samples), target coverage .95, and results are aggregated over 1000 replicates.}
%     \label{tab:mean}
% \end{table}

%\subsection{Pearson Correlation}\label{s:correlation}

\begin{table}[t]
    \centering
    \caption{Coverage and width of DP credible/confidence intervals for the mean of $[0,1]$ random variables, with Monte Carlo standard errors in parentheses. Simulation settings: $n=1{,}000$, $\epsilon_1=0.95$, $\epsilon_2=0.05$, $R=1{,}000$ (number of Monte Carlo samples), target coverage 0.95, results aggregated over 1{,}000 replicates.}
    \footnotesize % Reduce font size
    \renewcommand{\arraystretch}{1.0} % Tighter row height
    \setlength{\tabcolsep}{6pt} % Tighter column spacing
    \begin{tabular}{lll}
        \toprule
        \textbf{Method} & \textbf{Coverage} & \textbf{Width} \\
        \midrule
        \citet{wang2018statistical} (Known Variance) & 0.942 (0.007) & 0.023 ($0.0001$) \\
        Plug-In & 0.866 (0.011) & 0.022 (0.0003) \\
        PUMBA  & 0.952 (0.007) & 0.026 (0.0002) \\
        \bottomrule
    \end{tabular}
    \label{tab:mean}
\end{table}

\subsection{Simple Linear Regression by Sufficient Statistic Perturbation}\label{s:linearregression}

A common statistical task is to identify the dependence between two variables through linear regression. Several prior works have proposed methods for DP linear regression, however these each require at least one of the following: (i) a fully parametric model, (ii) specific priors (such as conjugate priors), (iii) significantly large sample sizes for asymptotic approximations to be accurate, and/or (iv) significant computational time/resources.  

A common approach to DP simple linear regression is through sufficient statistic perturbation. Assuming that the predictors and responses are bounded a priori (such as by prior domain knowledge), the sufficient statistics, which consist of the first, second, and mixed moments of the variables, have finite sensitivity. Given this setting, let $x_1,\ldots, x_n$ be an i.i.d. explanatory variable from a continuous distribution, and $y_i = \beta_0+\beta_1x_i+e_i$ be the response, where $e_i$ are i.i.d. with mean zero and finite variance $\sigma^2$. We assume that the data have been normalized (using prior knowledge) such that $x_i,y_i \in [0,1]$ for all $i=1,\ldots, n$. From the summary statistics, $\sum_{i=1}^n x_i$, $\sum_{i=1}^n x_i^2$, $\sum_{i=1}^n y_i$, $\sum_{i=1}^n y_i^2$, and $\sum_{i=1}^n x_iy_i$, the three parameters $\beta_0$ and $\beta_1$ can be estimated by least squares which is the best linear unbiased estimator (and asymptotically normal), whose distribution (conditional on the $x_i$'s) depends on $\sum_{i=1}^n x_i$, $\sum_{i=1}^n x_i^2$, and $\sigma^2$. Note that each of these statistics has sensitivity 1, so adding independent noise distributed as $N(0,5/\ep^2)$ satisfies $\ep$-GDP. Similar DP statistics have been used by \citet{bernstein2019differentially,awan2021structure,alabi2022hypothesis}. We propose using PUMBA to impute the summary statistics, from which we can estimate the three parameters, and then use the asymptotic distribution of $\hat\beta$ (with a plug-in unbiased estimator of $\sigma^2$) to approximate the variance. The constraints imposed on PUMBA consist of moment inequalities of the form in \eqref{eq:moments} for both the $x_i$'s and the $y_i$'s, as well as positive-definiteness for the ``$X^\top X$'' matrix (where $X$ includes a vector of ones and the explanatory variable), and non-negativity for the standard estimate of $\sigma^2$.

%\citet{awan2021structure} plugged these noisy sufficient statistics into the standard least squares estimator. \citet{bernstein2019differentially} showed that an approximate Bayesian inference procedure offered superior performance, but their Gibbs sampler (i) assumed a Gaussian model for the responses and (ii) used a Gaussian approximation for the sufficient statistics, which is not formally justified (unlike Theorem \ref{thm:fullTV}). 

 In Table \ref{tab:LR_sim}, we compare PUMBA with (i) a plug-in estimator (Naive), which treats noisy summary statistics as if they were original (as in \citet{awan2021structure}), and (ii) Bayesian inference through a mis-specified Gaussian model, via data augmentation MCMC (DA MCMC) from \citet{ju2022data}(see implementation details in Supplementary Materials, Section \ref{s:linDetails}). We consider $n\in\{30,100,300,1000\}$ with $\ep=1$ under the model $x_i\sim \mathrm{Unif}(0,1)$, $\beta_0=0.25$, $\beta_1=0.5$, and $e_i\iid (\mathrm{Beta}(4,4)-1/2)/2$ (mean zero, standard deviation $\sigma=1/12$). Results are averaged over 1,000 replicates. For PUMBA, we draw $R=1{,}000$ Monte Carlo samples per replicate. The DA MCMC uses a chain of length 10,000 with 5,000 burn-in iterations; effective sample size (ESS) is computed on the remaining samples, while runtime reflects the full chain (so ESS/sec would be roughly doubled without discarding the burn-in samples).

PUMBA exhibits slightly conservative but well-calibrated coverage that approaches the nominal level as $n$ increases, whereas the Naive method shows substantial under-coverage. In terms of RMSE, PUMBA consistently outperforms the plug-in estimator, with particularly large gains at small sample sizes; as $n$ grows, both methods converge due to their shared asymptotic distribution. Compared to DA MCMC under a misspecified Gaussian model, PUMBA maintains stable coverage as $n$ increases (while DA MCMC drops to 0.263 at $n=1{,}000$). Although DA MCMC achieves slightly smaller width and RMSE at $n=30$, PUMBA attains about $40\%$ lower RMSE at $n=1{,}000$. Moreover, PUMBA delivers substantial computational gains, achieving roughly $11$-times higher effective sampling rates at $n=30$ and up to $3{,}808$-times at $n=1{,}000$\footnote{Computations were run through interactive sessions on a cluster using 128 cores. The parallelization was run between replicates, so each ESS/sec relates to a single core. The processor was Amd Epyc 9575F.}. The latter increase in PUMBA's computational performance with larger $n$ is mainly due to the quality of the samples: for small $n$ the DP noise dominates leading to reject many samples that do not deliver statistics which respect the constraints, with this rejection rate therefore diminishing with larger $n$.

\begin{remark}
    While the coverage of the naive method could be potentially improved by incorporating the privacy noise into the sampling distribution, we would still expect a larger width for the naive CIs since the RMSE of the naive estimator is higher than that of the PUMBA posterior mean (when coverage is equated).
\end{remark}

\begin{remark}
    Note that the data augmentation MCMC uses a mis-specified Gaussian model and conjugate priors. Performance could be improved by using a better model, but we emphasize that (i) it may be challenging to develop a fully parametric model along with an appropriate prior distribution and (ii) the performance of the data augmentation MCMC suffers when using non-conjugate priors \citep[see][]{chen2025particle}.
\end{remark}

% \begin{table}[t]
% \centering
% \footnotesize
% \renewcommand{\arraystretch}{1.05}
% \setlength{\tabcolsep}{5pt}
% \caption{Simple linear regression inference for the slope. The number of simulations is $B=1{,}000$ and the nominal coverage is $0.95$.}
% \label{tab:LR_sim_combined}
% \begin{tabular}{c|ccc|cccc|cccc}
% \toprule
%  & \multicolumn{3}{c|}{\textbf{Naive}} 
%  & \multicolumn{4}{c|}{\textbf{DA MCMC}} 
%  & \multicolumn{4}{c}{\textbf{PUMBA}} \\
% \cmidrule(r){2-4} \cmidrule(r){5-8} \cmidrule(l){9-12}
% $\mathbf{n}$ 
% & Cov. & Width & RMSE 
% & Cov. & Width & RMSE & ESS/sec
% & Cov. & Width & RMSE & Samples/sec \\
% \midrule
% 30 
% & 0.095 (0.009) & 0.137 (0.010) & 59.642 (1.886)
% & 0.994 (0.002) & 1.164 (0.004) & 0.315 (0.010) & 71.934 (0.231)
% & 0.991 (0.003) & 2.515 (0.026) & 0.375 (0.012) & \phantom{x}843.251 (21.636) \\

% 100 
% & 0.154 (0.011) & 0.107 (0.004) & \phantom{x}0.881 (0.028)
% & 0.868 (0.011) & 0.653 (0.002) & 0.251 (0.008) & 16.350 (0.044)
% & 0.962 (0.006) & 1.285 (0.017) & 0.248 (0.008) & 2,538.516 (48.919) \\

% 300 
% & 0.218 (0.013) & 0.057 (0.001) & \phantom{x}0.112 (0.004)
% & 0.515 (0.016) & 0.330 (0.001) & 0.171 (0.005) & \phantom{x}4.483 (0.013)
% & 0.956 (0.006) & 0.413 (0.002) & 0.102 (0.003) & 3,833.302 (50.056) \\

% 1,000 
% & 0.391 (0.015) & 0.035 ($<0.001$) & \phantom{x}0.035 (0.001)
% & 0.263 (0.014) & 0.125 ($<0.001$) & 0.082 (0.003) & \phantom{x}1.414 (0.004)
% & 0.959 (0.006) & 0.131 ($<0.001$) & 0.033 (0.001) & 5,385.495 (23.492) \\
% \bottomrule
% \end{tabular}
% \end{table}

\begin{table}[t]
\centering
\footnotesize
\renewcommand{\arraystretch}{1.05}
\setlength{\tabcolsep}{6pt}
\caption{Coverage, width of DP credible/confidence intervals, RMSE and ESS/sec for simple linear regression inference for the slope, with Monte Carlo standard errors in parentheses. Simulation settings: $\epsilon=1$, $R=1{,}000$ (number of Monte Carlo samples for PUMBA), target coverage 0.95, results aggregated over 1{,}000 replicates.}
\label{tab:LR_sim}
\begin{tabular}{c c c c c c}
\toprule
\textbf{Method} & $\mathbf{n}$ & \textbf{Coverage} & \textbf{Width} & \textbf{RMSE} & \textbf{ESS/sec} \\
\midrule

\multirow{4}{*}{Plugin}
& 30   & 0.095 (0.009) & 0.137 $\phantom{<}$(0.010) & 59.642 (1.886) & -- \\
& 100  & 0.154 (0.011) & 0.107 $\phantom{<}$(0.004) & \phantom{x}0.881 (0.028) & -- \\
& 300  & 0.218 (0.013) & 0.057 $\phantom{<}$(0.001) & \phantom{x}0.112 (0.004) & -- \\
& 1,000& 0.391 (0.015) & 0.035 ($<0.001$) & \phantom{x}0.035 (0.001) & -- \\

\midrule

\multirow{4}{*}{DA MCMC}
& 30   & 0.994 (0.002) & 1.164 $\phantom{<}$(0.004) & 0.315 (0.010) & \phantom{x}71.934 (0.231) \\
& 100  & 0.868 (0.011) & 0.653 $\phantom{<}$(0.002) & 0.251 (0.008) & \phantom{x}16.350 (0.044) \\
& 300  & 0.515 (0.016) & 0.330 $\phantom{<}$(0.001) & 0.171 (0.005) & \phantom{xx}4.483 (0.013) \\
& 1,000& 0.263 (0.014) & 0.125 ($<0.001$) & 0.082 (0.003) & \phantom{xx}1.414 (0.004) \\

\midrule

\multirow{4}{*}{PUMBA}
& 30   & 0.991 (0.003) & 2.515 $\phantom{<}$(0.026) & 0.375 (0.012) & \phantom{x}843.251 (21.636) \\
& 100  & 0.962 (0.006) & 1.285 $\phantom{<}$(0.017) & 0.248 (0.008) & 2,538.516 (48.919) \\
& 300  & 0.956 (0.006) & 0.413 $\phantom{<}$(0.002) & 0.102 (0.003) & 3,833.302 (50.056) \\
& 1,000& 0.959 (0.006) & 0.131 ($<0.001$) & 0.033 (0.001) & 5,385.495 (23.492) \\

\bottomrule
\end{tabular}

\end{table}

\section{Application: Drivers of Homeownership}
\label{s:homeownership}

To highlight the effectiveness of this approach, we apply it to a regression task of interest in Political Science. Specifically, we study the drivers of homeownership across U.S. counties, focusing on the role of ethnic composition and key socio-economic factors such as unemployment, poverty, health insurance coverage, house cost and population density. These variables are commonly used to understand barriers to homeownership and inform policy decisions. While these variables are important, many are sensitive, motivating the use of differential privacy. The American Community Survey (ACS), conducted annually by the U.S. Census Bureau, provides detailed demographic, social and economic data across geographic levels resulting in $3,222$ counties and a total population over $334$ million individuals. Using ACS 5-year estimates, we fit a standard linear model to study homeownership since this has been a widely employed approach for analyzing ACS data in general \citep{stone2015place, frankenfeld2019county, allaire2018national, mueller2021widespread, parilla2018examining, walker2011variegated}. We refer the reader to Section \ref{s:lit_rev_home} in Supplementary Materials for additional literature review on this topic. For this particular study we consider the following linear model:
\begin{align}
        \mathbb{E}[\text{Home Ownership}] =& \,\beta_0 + \sum_{i=1}^{4} \beta_i \, \text{Ethnicity}_i + \beta_5 \, \text{Unemployed} + \beta_6 \, \text{Poverty} + \nonumber\\
        & \beta_7 \, \text{No Insurance} + \beta_8 \, \text{Housecost} + \beta_9 \, \text{Pop. Density} ,
        \label{eq.homeownership}
\end{align}
where $\mathbb{E}[\text{Home Ownership}]$ is the expected proportion of homeowners in a county. For the purpose of this example, we consider five general ethnic backgrounds of interest as potential drivers: ``White'', ``Black'', ``Asian'', ``Indigenous'' (which includes American Indian and Alaska Native plus Native Hawaiian and Other Pacific Islander) and ``Other''. We use the ``White'' group as a reference level included in the intercept.

This setting provides a useful testbed for PUMBA: since ACS data are not currently privatized, we can simulate differential privacy mechanisms and compare results obtained from privatized data to those from the original data, allowing us to assess the reliability of our approach. For this reason, we run the Gaussian additive mechanism to the county statistics $T_k$ to guarantee DP and select the privacy budget $\epsilon = 0.25$. The overall sensitivity for each county is $\Delta = \sqrt{8}$, since the ethnicity variable can be considered a histogram with sensitivity $\sqrt{2}$, implying that the privacy budget can be rescaled to $\epsilon^* = \nicefrac{\epsilon}{\sqrt{8}}$.

% \textcolor{red}{We want to fit equation linear regression based on these DP summaries. Previous solutions are not satisf. or inapplicable: (i) sample size (more than 334 million); (ii) require fully parametric models for all variables }

As stated earlier, our goal is now to perform linear regression (and inference) using the DP summaries described above. Unfortunately, as mentioned in the introduction to this work, there are no currently available methods to perform statistical inference on this type of problem, i.e. (i) a large sample size collecting more than 334 million individuals and (ii) a dimension and variety of considered variables for which it would be complicated to assume a strong parametric form. Therefore, we study how the proposed method (PUMBA) compares to the original non-private regression (Non-DP) and to a naive approach (Naive) which simply runs the regression on the privatized data and performs inference without considering the noise added for privacy. In particular, since we are using the Gaussian additive mechanism, we can more reasonably assume an approximately normal posterior $\theta\,|\, T_n$ and therefore use the alternative sampling algorithm discussed in Remark \ref{rmk:alt_algo} and detailed in Algorithm 2 in the Supplementary Materials. 

To study these approaches, we firstly run an emulation study (Section \ref{s.home_sim}) inspired from the ACS data to evaluate the performance of these methods in a controlled setting (similar in spirit to the numerical experiments in Section \ref{s:numerical_exp}), and then we study how its conclusions compare to the other approaches when actually applied to the ACS data (Section \ref{sec.acs_analysis}). In the latter analysis we use two sampling techniques for the count variables: (i) sample from their exact discrete support and (ii) sample from a continuous support (as a form of approximation). Noting no difference in results between these approaches for this data, we report those based on the exact discrete sampling for the original ACS data analysis, while in the emulation study we sample from a continuous support to speed-up computations.

\subsection{Homeownership Emulation Study}
\label{s.home_sim}

For this first study we estimate the regression model in \eqref{eq.homeownership} on the original non-private ACS data whose estimated coefficients are presented in the first column of Table \ref{tab:regression_results} further on. We take these estimated coefficients as the true coefficients to be used in the following simulation settings:
\begin{enumerate}
    \item $H_0$: simulation under the null hypothesis where $\beta_i = 0$, for all $i = 1, \hdots, p$, i.e. only the intercept $\beta_0$ is significant.
    \item $H_A$: simulation under the alternative hypothesis where all $\beta_i \neq 0$, for all $i = 0, \hdots, p$, i.e. all variables have a significant impact on the response.
\end{enumerate}
We simulate from the linear model $B = 1,000$ times and, in each case, we compute confidence/credible intervals for the three considered methods. For the PUMBA we use $R = 1,000$ (number of Monte Carlo samples). While we expect the Non-DP to attain the nominal coverage, we expect the Naive method to slightly undercover while we would want PUMBA to at least attain the nominal coverage. %Indeed, in general, we would prefer to have conservative intervals rather than have a large type-I error. 
Table \ref{tab:coverage_results} reports the coverage of $1 - \alpha$ confidence/credible intervals (with $\alpha = 0.05$) for all three considered methods.

\begin{table}[t]
    \centering
    \scriptsize
    \sisetup{scientific-notation=true, round-mode=figures, round-precision=3}
    \renewcommand{\arraystretch}{1.05}
    \setlength{\tabcolsep}{4pt}
    \caption{Coverage probabilities and average interval widths for Non-DP, Naive, and PUMBA under the null ($H_0$) and alternative hypotheses ($H_A$). The number of simulations is $B = 1{,}000$ and the nominal coverage is $1 - \alpha = 0.95$.}
    \label{tab:coverage_results}
 \resizebox{\textwidth}{!}{%
    \begin{tabular}{l
        cc cc cc 
        cc cc cc}
        \toprule
         & \multicolumn{6}{c}{$H_0$} & \multicolumn{6}{c}{$H_A$} \\
        \cmidrule(lr){2-7} \cmidrule(lr){8-13}
         & \multicolumn{2}{c}{Non-DP} 
         & \multicolumn{2}{c}{Naive} 
         & \multicolumn{2}{c}{PUMBA} 
         & \multicolumn{2}{c}{Non-DP} 
         & \multicolumn{2}{c}{Naive} 
         & \multicolumn{2}{c}{PUMBA} \\
         & Cov. & Width & Cov. & Width & Cov. & Width 
         & Cov. & Width & Cov. & Width & Cov. & Width \\
        \midrule
        Intercept    & 0.959 & \num{1.444e-3} & 0.565 & \num{3.494e-3} & 1.000 & \num{1.763e-2} & 0.959 & \num{1.444e-3} & 0.262 & \num{4.058e-3} & 0.999 & \num{6.320e-2} \\
        Black        & 0.944 & \num{2.387e-3} & 0.817 & \num{5.868e-3} & 1.000 & \num{1.321e-2} & 0.944 & \num{2.387e-3} & 0.682 & \num{6.884e-3} & 1.000 & \num{2.445e-2} \\
        Asian        & 0.958 & \num{1.138e-2} & 0.625 & \num{2.788e-2} & 0.999 & \num{1.021e-1} & 0.958 & \num{1.138e-2} & 0.443 & \num{3.318e-2} & 1.000 & \num{1.463e-1} \\
        Indigenous   & 0.965 & \num{4.473e-3} & 0.740 & \num{1.100e-2} & 1.000 & \num{2.717e-2} & 0.965 & \num{4.473e-3} & 0.712 & \num{1.286e-2} & 1.000 & \num{3.557e-2} \\
        Other        & 0.942 & \num{3.656e-3} & 0.710 & \num{9.001e-3} & 1.000 & \num{2.356e-2} & 0.942 & \num{3.656e-3} & 0.662 & \num{1.056e-2} & 1.000 & \num{4.480e-2} \\
        Unemployed   & 0.948 & \num{3.267e-2} & 0.473 & \num{7.817e-2} & 0.998 & \num{4.467e-1} & 0.948 & \num{3.267e-2} & 0.445 & \num{9.151e-2} & 1.000 & \num{5.370e-1} \\
        Poverty      & 0.948 & \num{8.785e-3} & 0.609 & \num{2.148e-2} & 0.999 & \num{6.669e-2} & 0.948 & \num{8.785e-3} & 0.587 & \num{2.512e-2} & 1.000 & \num{9.187e-2} \\
        No Insurance & 0.947 & \num{6.804e-3} & 0.749 & \num{1.673e-2} & 1.000 & \num{4.251e-2} & 0.947 & \num{6.804e-3} & 0.727 & \num{1.958e-2} & 1.000 & \num{6.571e-2} \\
        Housecost    & 0.955 & \num{4.660e-4} & 0.475 & \num{1.101e-3} & 1.000 & \num{6.457e-3} & 0.955 & \num{4.660e-4} & 0.203 & \num{1.267e-3} & 0.998 & \num{2.574e-2} \\
        Pop. Density & 0.958 & \num{1.722e-7} & 0.965 & \num{4.247e-7} & 0.999 & \num{7.335e-7} & 0.958 & \num{1.722e-7} & 0.865 & \num{7.187e-7} & 1.000 & \num{9.110e-7} \\
        \bottomrule
    \end{tabular}
    }
\end{table}
We can observe that the method in absence of DP noise (i.e. Non-DP) does indeed have the expected levels of coverage. However, if one were to ignore the noise introduced to guarantee DP (i.e. Naive), then the true coefficients are severely undercovered in all settings (with the exception of the last coefficient related to population density). On the other hand, the proposed PUMBA is significantly conservative for all coefficients, therefore guaranteeing a low type-I error. The latter of course comes at the cost of statistical power: while the Non-DP has 100\% power for the considered sample size, the Naive and the PUMBA also are generally close to 100\% power for all coefficients except for ``Unemployed'' where the power however is much lower for the PUMBA compared to the Naive. Having observed the behavior of the three methods under controlled settings, the next section discusses how these three differ when directly analyzing the real data.

\subsection{ACS Study on Homeownership}
\label{sec.acs_analysis}

The results presented above support the validity of PUMBA in performing adequate inference when data is privatized, as done for example by the U.S Census through DP mechanisms. Having confirmed this, we now study how the considered methods align or differ when directly using the ACS data to understand the drivers of homeownership considered in the model in \eqref{eq.homeownership}. Ideally, we would like the proposed method to align with the Non-DP approach (i.e. standard regression on original non-privatized data) while we expect that the conclusions of the Naive approach (i.e. standard regression on privatized data) to potentially deliver different conclusions since it does not account for the privacy mechanism $\mathcal{M}$. The results of the analysis using these three approaches are  in Table \ref{tab:regression_results}.

\begin{table}[t]
    \centering
    \footnotesize % Smaller font to reduce width
    \renewcommand{\arraystretch}{1.0} % Reduce row height
    \setlength{\tabcolsep}{6pt} % Reduce horizontal padding between columns
    \begin{threeparttable}
    \caption{Coefficient Estimates and corresponding p-values (in parentheses) for Non-DP approach (first column); Naive approach (second column) and PUMBA (third column). Significance levels: * $< 0.05$, ** $< 0.01$, *** $< 0.001$.}
    \label{tab:regression_results}
    \begin{tabular}{l
        S[table-format=1.3]@{\,}c
        S[table-format=1.3]@{\,}c
        S[table-format=1.3]@{\,}c}
        \toprule
         & \multicolumn{2}{c}{Non-DP} & \multicolumn{2}{c}{Naive} & \multicolumn{2}{c}{PUMBA} \\
        & {Estimate} & {(p-value)} & {Estimate} & {(p-value)} & {Estimate} & {(p-value)} \\
        \midrule
        Intercept  & 0.738 & ($<$0.001) *** & 0.743 & ($<$0.001) *** & 0.753 & ($<$0.001) *** \\
        Black      & -0.120 & ($<$0.001) *** & -0.121 & ($<$0.001) *** & -0.126 & ($<$0.001) *** \\
        Asian      & -1.030 & ($<$0.001) *** & -1.026 & ($<$0.001) *** & -1.000 & ($<$0.001) *** \\
        Indigenous & -0.108 & ($<$0.001) *** & -0.104 & ($<$0.001) *** & -0.109 & ($<$0.001) *** \\
        Other      & -0.115 & ($<$0.001) *** & -0.111 & ($<$0.001) *** & -0.107 & ($<$0.001) *** \\
        Unemployed & -0.095 & (0.458) & -0.303 & (0.013) * & -0.282 & (0.158) \\
        Poverty    & -0.402 & ($<$0.001) *** & -0.371 & ($<$0.001) *** & -0.384 & ($<$0.001) *** \\
        No Insurance & -0.040 & (0.136) & -0.048 & (0.071) . & -0.047 & (0.11) \\
        Housecost  & 0.030 & ($<$0.001) *** & 0.029 & ($<$0.001) *** & 0.025 & (0.008) ** \\
        Pop. Density & \num{-8.17e-6} & ($<$0.001) *** & \num{-8.13e-06} & ($<$0.001) *** & \num{-8.21e-06} & ($<$0.001) *** \\
        \bottomrule
    \end{tabular}
    \end{threeparttable}
\end{table}

Overall, the three approaches agree on the direction and size of the effects of all variables on the proportion of homeowners within each county. All considered drivers appear to have a negative impact implying that, intuitively, higher unemployment and poverty levels decrease the proportion of homeowners in a county, as do a lower proportion of insured individuals and a higher population density. A similar conclusion can be made for the different ethnic groups when compared to the reference group of ``White''. The only exception is the housecost which, counterintuitively, appears to increase the proportion of homeowners: this however can be considered as a proxy for income levels, thereby suggesting the higher income counties see an increased proportion of homeowners. Putting aside the interpretation of the outputs, a major difference between the considered approaches can be seen however in the effect of the unemployment rate which sees a much greater size for the Naive and PUMBA approaches with respect to the Non-DP original results. Nevertheless, this driver is not statistically significant in the Non-DP analysis and the PUMBA also comes to the same conclusion. On the other hand, the Naive approach reports a significant impact of this variable on homeownership thereby bringing to different conclusions compared to the reference Non-DP method. This is also the case for the significance of insurance coverage on homeownership, where Non-DP and PUMBA agree on its non-significant effect while, although not at the typical $\alpha = 0.05$ level, the Naive approach detects a reasonable level of significance of this variable. Therefore, the proposed PUMBA delivers the same conclusions in terms of statistical significance as the reference non-private analysis highlighting how, unlike the Naive approach, it is able to appropriately account for privacy noise when performing statistical inference. While linear regression is a common tool to address these problems in Political Sciences, logistic regression (weighted by county population) is another reasonable model for this kind of problem. The results of this analysis (and corresponding simulation study as in Section \ref{s.home_sim}) can be found in Appendix \ref{s.acs_study_logistic} where all methods obtain similar estimates and inferential conclusions, with the Naive approach giving slight undercoverage while PUMBA still remains reliable with coverages close to the nominal level.

\section{Discussion}
\label{s:discussion}

PUMBA provides a relatively broad and computationally feasible option for researchers and practitioners who aim to perform inference based on privatized statistics (typically counts) reported in numerous data sources such as, for example, national surveys like the U.S. Census or related products. With strong theoretical guarantees under generally mild assumptions, PUMBA contributes to the avenue of research focused on reliable statistical inference under privacy constraints by also overcoming some computational bottlenecks common to various current solutions from the Bayesian framework. This is achieved by justifying and proving the validity of different posterior approximations for the quantities of interest when working with statistics computed on large amounts of observations. Indeed, these asymptotic approximations allow to deliver more efficient sampling schemes to obtain posterior samples. The validity of the theoretical results are confirmed through some simulation studies, including inference for the mean as well as linear and logistic regression. Moreover, an applied study of interest in the field of Political Science highlights how PUMBA achieves similar conclusions to the non-private analysis as opposed to an approach that would not consider the randomness due to privacy.

Considering limitations and possible improvements, if the privatized statistics are not computed on large enough samples, then PUMBA may not be a good solution since it relies on asymptotic approximations which would not hold in this case. However, based on our experiments, it would appear that PUMBA tends to be conservative in small sample scenarios. Another aspect concerns the use of different types of DP mechanisms: indeed, while other types of DP mechanisms can be employed within PUMBA, currently the theoretical guarantees can only be achieved when privatizing the data through additive mechanisms. Finally, an implicit requirement for our sampling procedure is for the privacy budget $\epsilon$ to be reasonably small: this allows the proposed posterior sampling to be manageable from a feasible target support. While the asymptotic framework cannot be removed given the nature of the method, other limitations can be addressed in future work, starting from weakening assumptions such as those on the form of the DP mechanism as well as the exploration of more efficient sampling algorithms (e.g. random walk sampling).

The theory developed in this paper could be potentially adapted to justify similar approximations in other sampling schemes. For example, \citet{bernstein2018differentially,bernstein2019differentially} substitute a Gaussian central limit theorem approximation in their Gibbs sampler without formal justification. We conjecture that using a local limit theorem, along with an adaptation of our analyses, could potentially give theoretical support to their procedure. 

%\textcolor{blue}{Jordan: add discussion on alternative with Gibbs style sampler.}

% {\color{red}
% - SHORT RECAP
% - LIMITATIONS (asymptotic, mechanism, rejection sampling with large epsilon (small target))
% - FUTURE WORK (weakening conditions?, better sampler e.g. random walk depending on mechanism, other applications to sensitive data, to linear regression + noise)
% }

\section*{Acknowledgments}
The clusters at the Center for Research Computing and Data at the University of Pittsburgh were used to support some of the simulation studies, primarily those in Section 4.2. 
ChatGPT 5.3 was used as a support to clean and comment code included in the supplementary material, as well as to shorten or rephrase a few portions of the manuscript. The authors have reviewed all the material and take full responsibility for its content.

\bigskip
\begin{center}
{\large\bf SUPPLEMENTARY MATERIAL}
\end{center}

\begin{description}
%\item[Title:] Brief description. (file type)

\item[supplement.pdf] Document with proofs, technical details, and ancillary results. 

\item[supplement.zip] R code to reproduce the analyses in this paper.
\end{description}

\bigskip
\begin{center}
{\large\bf DATA AVAILABILITY STATEMENT}
\end{center}

The authors confirm that the data supporting the findings of this study are available within the article and its supplementary materials. The original data can also be found from the following resources available in the public domain: \url{https://www.census.gov/programs-surveys/acs/news/data-releases/2022/release.html#fiveyear}

\newpage

\appendix
\section{Proofs and Technical Details}
\label{app:proofs}

In this appendix we provide all details and preliminary results needed to prove the final results presented in the main manuscript. To facilitate understanding in a self-contained manner, we explicitly state conditions needed for these results without referring to the assumptions listed in the main manuscript, except where appropriate to do so.

\subsection{Proof of Theorem 1}

To prove Theorem 1, we firstly need some preliminary results. Let us start with the following lemma which is inspired by \citet{bioche2016approximation}, who study improper priors as a sequence of proper priors. However, while their results only give convergence in distribution, we are able to assert point-wise convergence of the posterior densities, which implies convergence in total variation distance. 
%\begin{comment}
\begin{lemma}\label{lem:seqPrior}
    Let $p_n$ be a sequence of densities on $\RR^k$, let $m(\sdp\,|\,T_n)$ be a conditional density that does not depend on $n$, both with respect to a base measure $\lambda$, and let $\sdp$ be held fixed. Suppose there exists a sequence of positive values $(a_n)_{n=1}^\infty$ such that $a_n p_n$ converges pointwise to a non-zero function $\tilde{p}:\RR^k\rightarrow \RR^{\geq 0}$, where $m(\sdp\,|\, T_n)\,\tilde{p}(T_n)$ is integrable with respect to $\lambda$. Suppose  further that there exists a function $g:\RR^k\rightarrow \RR^{\geq 0}$ such that $a_n p_n(T_n) \leq g(T_n)$ for all $T_n$ and $m(\sdp\,|\, T_n)\,g(T_n)$ is integrable with respect to $\lambda$. Then, the posterior distribution 
    $$p_n(T_n\,|\, \sdp) = \frac{m(\sdp\,|\, T_n)\,p_n(T_n)}{\int m(\sdp\,|\,T_n)\,p_n(T_n)\,d\lambda(T_n)},$$
    converges pointwise to 
    $$p_\infty (T_n\,|\,\sdp) = \frac{m(\sdp\,|\,T_n)\, \tilde{p}(T_n)}{\int m(\sdp\,|\,T_n)\,\tilde{p}(T_n)\,d\lambda(T_n)}.$$ 
    It follows that
    $$\mathrm{TV}(p_n(T_n\,|\,\sdp), p_\infty (T_n\,|\,\sdp)) \rightarrow 0.$$ 
\end{lemma}
\begin{proof}
    We can rewrite the sequence of posterior distributions as 
    $$p_n(T_n\,|\,\sdp) = \frac{m(\sdp\,|\,T_n) \, a_n p_n(T_n)}{\int m(\sdp\,|\,T_n=y)\, a_n p_n(T_n=y)\,d\lambda(y)}.$$
    The numerator converges pointwise to $m(\sdp\,|\,T_n) \, \tilde{p}(T_n)$ by assumption and the denominator converges pointwise to $\int m(\sdp\,|\,T_n=y)\,\tilde{p}(T_n=y)\,d\lambda(y)$ by the dominated convergence theorem. By Scheff\'ees lemma (Corollary 2.30, \citealp{VanDerVaart2000}), this implies convergence in total variation.
\end{proof}
Note that Lemma \ref{lem:seqPrior} does not require that $\tilde{p}$ integrates to 1. In fact, we will be primarily interested in the case that $\tilde{p}$ is an ``improper distribution''. We now report an existing result which will be needed for our following proofs.
%\end{comment}
\begin{lemma}[Corollary 1: \citealp{shervashidze2006convergence}]\label{lem:shervashidze}
   Let $(u_{n,1},\ldots, u_{n,n})_{n=1}^\infty$ be a triangular array of random vectors in $\RR^k$ such that $u_{n,1},\ldots, u_{n,n}$ are independent for all $n$. Let $S_n=\sum_{j=1}^n u_{n,j}$ and let $p_{S_n}$ denote the density of $S_n$. Suppose that: (i) $u_{n,j}$ have densities (with respect to Lebesgue) bounded above by $An^{k/2}$ for some constant $A>0$; (ii) $\EE u_{n,j}=0$; (iii) $\EE S_n S_n^\top=I$; and (iv) $\sigma^2_{n,j}\defeq \EE \lVert u_{n,j} \rVert^2$ are upper bounded by $\sigma^2_{n,j} \leq \frac{k+1}n$. Then, if $S_n$ converges in distribution to $N_k(0,I)$, we have that 
   \[\sup_v|p_{S_n}(v)-\phi(v)| \overset p \rightarrow 0,\]
   where $\phi$ is the density of $N_k(0,I)$.
   \end{lemma}

In our case, $S_n=T_n$ is our statistic of interest. Based on this lemma, we will now deliver two additional auxiliary lemmas which apply Lemma \ref{lem:shervashidze} to the settings most useful for this work. More specifically, following the above result, we deliver a local limit theorem for continuous vectors (first lemma below) and then report a similar result for integer-valued vectors (second lemma below). 

\begin{lemma}[Local Limit Theorem for Continuous Vectors]\label{lem:lltcont}
   Let $x_1,x_2,\ldots, x_n$ be i.i.d. random vectors in $\RR^k$ such that: (i) their distribution has a density function which is bounded; (ii) $\EE x_i=\mu$, and $\EE(x_i-\mu)(x_i-\mu)^\top=\Sigma$ which has finite entries and is positive definite. Let $S_n = n^{-1/2}\Sigma^{-1/2}\sum_{i=1}^n (x_i-\mu)$ which has density $p_{S_n}$. Then,
   \[\sup_v |p_{S_n}(v)-\phi(v)| \overset p \rightarrow 0,\]
   where $\phi$ is the density of $N_k(0,I)$.
\end{lemma}
\begin{proof}
    Let us define the elements of a triangular array by $x_{n,j}=n^{-1/2}\Sigma^{-1/2} (x_j-\mu)$, for all $n=1,2,\ldots$ and $j=1,2,\ldots, n$. Then, $S_n = \sum_{j=1}^n x_{n,j}$. Firstly, note that (i) if the density of $x_i$ is bounded above by $A$, then the density of $x_{n,j}$ is bounded above by $A\det(\Sigma^{1/2})n^{k/2}$ (by change of variables); (ii) $\EE x_{n,j}=0$, (iii) $\EE S_n S_n^\top = I$, and $\sigma^2_{n,j}=\EE \lVert x_{n,j}\rVert^2=\frac{k}{n}\leq\frac{k+1}{n}$. By the central limit theorem (Theorem B, Section 1.9.1, \citealp{serfling1980approximation}), we know that $S_n$ converges in distribution to $N_k(0,I)$. Hence, we have satisfied the conditions of Lemma \ref{lem:shervashidze} and the result follows.
\end{proof}

We now consider the case of integer-valued vectors for which we found the theorem below reported in \citet{gamkrelidze2015local}. This theorem was attributed to \citet{meizler1949multidimensional} which, unfortunately, we were not able to locate directly.

\begin{lemma}[Local Limit Theorem for Integer-Valued Vectors:  \citealp{meizler1949multidimensional,gamkrelidze2015local}]\label{lem:lltint}
   Let $x_1,x_2,\ldots, x_n$ be a sequence of i.i.d. random vectors in $\ZZ^k$. Call $S_n = \sum_{i=1}^n x_i$, which has pmf $p_{S_n}$. Suppose that $\EE x_i=\mu$ and that $\Sigma_n = \EE(S_n-n\mu)(S_n-n\mu)^\top$ is positive definite. If $e_1,\ldots, e_k$ is a basis for a minimal lattice containing the support of $x_i$ and $\det(e_1,\ldots, e_k)=1$, then
   \[\sup_v |p_{S_n}(v)-\phi(v;n\mu,\Sigma_n)| \overset p \rightarrow 0,\]
   where $\phi(v;\mu,\Sigma_n)$ is the pdf of $N_k(n\mu,\Sigma_n)$ evaluated at $v$, and the supremum is over the lattice points $v=\sum_{i=1}^k a_i e_i$ for $a_i \in \ZZ$. 
\end{lemma}

With the above results, we now have all we need to prove Theorem 1.

%\thmTgivenS*
\begin{proof}[Proof of Theorem 1.]

    We will establish the frequentist result first, which has three steps to the proof. We first derive upper and lower bounds for $p_n(T_n\,|\,\sdp,\theta)$ in the continuous and discrete case separately, then we derive the total variation result in the case that the prior predictive mean is zero, and finally we extend the result to arbitrary prior predictive mean values.\\

    Let $\theta\in \Theta$ be fixed. We write $\EE_\theta$ as abbreviation for either $\EE_{t\sim p(t\,|\,\theta)}$ or $\EE_{T_n\sim p_n(T_n\,|\,\theta)}$ as context should make clear. Denote by  $\Sigma_\theta$, the covariance matrix for the distribution $p(t\,|\,\theta)$. 
   \smallskip  
   
    \noindent \textbf{Step 1} For this step, we will suppose that $\EE_\theta t_i=0$ and address the case $\EE_\theta t(x_i)\neq 0$ in step 2. The objective of this step is to first apply a local limit theorem to approximate $p_n(T_n\,|\,\theta)$ as a Gaussian, and then bound the magnitude of $p_n$ in terms of the magnitude of the associated Gaussian. Ultimately this will allow us to bound the contribution of $p_n$ in the distribution $p_n(T_n\,|\,\sdp,\theta)$. While analogous, we tackle the continuous and discrete cases separately. \\
    
    \noindent \textit{Continuous Case}: Let us assume that $\lambda$ is Lebesgue measure on $\RR^k$ and that $p(t\,|\,\theta)$ is a bounded density (as per Assumption \ref{a.bounded_density}). Then we have that $p_n(n^{1/2} \Sigma^{1/2}_\theta u\,|\,\theta)n^{1/2} \det(\Sigma_\theta)^{1/2}$ is the density of $n^{-1/2} \Sigma^{-1/2}_\theta T_n$ which, by Lemma \ref{lem:lltcont}, converges uniformly to $\phi(u)$, i.e. the density of $N_k(0,I)$. Now, let $a_n=n^{1/2} \det(\Sigma_\theta)^{1/2}$ (suppressing the dependence on $\theta$ for simplicity of presentation) which allows us to define $a_n p_n(v)=n^{1/2}\det(\Sigma_\theta)^{1/2}p_n(v\,0\,|\theta)$. Then, with $\delta>0$, there exists $N$ (also possibly depending on $\theta$) such that for all $n\geq N$, 
    \[\sup_v \Big|n^{1/2} \det(\Sigma_\theta)^{1/2}p_n(v\,|\,\theta)-\phi(n^{-1/2} \Sigma^{-1/2}_\theta v)\Big|<\delta,\] by substituting $v=n^{1/2}\Sigma^{1/2}_\theta u$. So, for $n\geq N$, we have that for every $T_n$, $\phi(n^{-1/2}\Sigma^{-1/2}_\theta T_n)-\delta \leq a_np_n(T_n\,|\,\theta)\leq \phi(n^{-1/2}\Sigma^{-1/2}_\theta  T_n) + \delta$.
    For any fixed $T_n$, the lower bound converges to $(2\pi)^{-k/2} - \delta$ and the upper bound converges to $(2\pi)^{-k/2} + \delta$. Since $\delta$ was chosen arbitrarily, we see that $a_n p_n(T_n\,|\,\theta)$ converges pointwise to $(2\pi)^{-k/2}$. This implies that the support $\mathcal T_n$ converges to the support of $\lambda$. Incorporating the indicator for the support $\mathcal T_n$, we can also say that $$|a_np_n(T_n\,|\,\theta)-(2\pi)^{-k/2}I(T_n\in\mathcal T_n)|\leq |a_np_n(T_n\,|\,\theta)-(2\pi)^{-k/2}|+(2\pi)^{-k/2}|1-I(T_n\in \mathcal T_n)|\rightarrow 0$$.

     Furthermore, since the standard multivariate Gaussian density is maximized at zero, achieving the value $(2\pi)^{-k/2}$, we have for $n\geq N$ that 
     \[a_n p_n(T_n\,|\,\theta)\leq (2\pi)^{-k/2}I(T_n\in \mathcal T_n)(2\pi)^{-k/2}+\delta\efdeq g(T_n),\] where we introduced the indicator variable to give a tighter upper bound. Since $g(T_n)$ is a constant, and we assume that $m(\sdp\,|\,T_n)$ is integrable, so is $m(\sdp\,|\,T_n)g(T_n)$. By the dominated convergence theorem (and taking absolute value), we have that 
    \[\left|\int_{\mathcal T_n} m(\sdp\,|\,T_n=y)a_n p_n(T_n=y\,|\,\theta)\ d\lambda(y) - \int_{\mathcal T_n} m(\sdp\,|\,T_n=y) (2\pi)^{-k/2}\ d\lambda(y)\right|\rightarrow 0.\]

    \noindent \textit{Discrete Case}: By Lemma \ref{lem:lltint}, 
    \[\sup_v \big|p_n(v\,|\,\theta)-\phi(v;0,n\Sigma_\theta)\big| \overset p \rightarrow 0,\]
    where $\phi(v;0,n \Sigma_\theta)$ is the pdf of $N_k(0,n\Sigma_\theta)$ evaluated at $v$. Then, setting $a_n = n^{1/2}\det(\Sigma_\theta)^{1/2}$ we have $a_np_n(v\,|\,\theta)=n^{1/2}\det(\Sigma_\theta)^{1/2}p_n(v)$. Letting $\delta>0$, then there exists $N$ such that for all $n\geq N$, 
    \[\sup_v \big|n^{1/2}\det(\Sigma_\theta)^{1/2}p_n(v\,|\,\theta)-n^{1/2}\det(\Sigma_\theta)^{1/2}\phi(v;0,n\Sigma_\theta)\big| < \delta.\]
    So, for $n\geq N$, we have that for every $T_n$,
    \[n^{1/2}\det(\Sigma_\theta)^{1/2}\phi(T_n;0,n\Sigma_\theta)-\delta 
    \leq a_n p_n(T_n\,|\,\theta) \leq n^{1/2}\det(\Sigma_\theta)^{1/2}\phi(T_n;0,n\Sigma_\theta)+\delta.\]
    For any fixed $T_n$, the lower bound converges to $(2\pi)^{-k/2} - \delta$ and the upper bound converges to $(2\pi)^{-k/2} + \delta$. Since $\delta$ was chosen arbitrarily, we see that $a_n p_n(T_n\,|\,\theta)$ converges pointwise to $(2\pi)^{-k/2}$. This implies that the support $\mathcal T_n$ converges to the support of $\lambda$. Incorporating the indicator for the support $\mathcal T_n$, we can also say that $|a_np_n(T_n\,|\,\theta)-(2\pi)^{-k/2}I(T_n\in\mathcal T_n)|\leq [a_np_n(T_n\,|\,\theta)-(2\pi)^{-k/2}]+(2\pi)^{-k/2}[1-I(T_n\in \mathcal T_n)]\rightarrow 0$.

     Furthermore, $n^{1/2}\det(\Sigma_\theta)^{1/2}\phi(T_n;0,n\Sigma_\theta)$ is maximized at zero, achieving the value $(2\pi)^{-k/2}$, we have for $n\geq N$ that 
     \[a_n p_n(T_n\,|\,\theta)\leq (2\pi)^{-k/2}I(T_n\in \mathcal T_n)\leq(2\pi)^{-k/2}+\delta\efdeq g(T_n),\] where we introduced the indicator variable to give a tighter upper bound. Since $g(T_n)$ is a constant, and we assume that $m(\sdp\,|\,T_n)$ is summable, so is $m(\sdp\,|\,T_n)g(T_n)$. By the dominated convergence theorem (and taking absolute value), we have that 
    \[\left|\int_{\mathcal T_n} m(\sdp\,|\,T_n=y)a_n p_n(T_n=y\,|\,\theta)\ d\lambda(y) - \int_{\mathcal T_n} m(\sdp\,|\,T_n=y) (2\pi)^{-k/2}\ d\lambda(y)\right|\rightarrow 0.\]

    \noindent \textbf{Step 2}: We tackle the continuous and discrete cases simultaneously, allowing for $\lambda$ to be either base measure. Given $0<\delta<(2\pi)^{-k/2}$, there exists $N$ such that for $n\geq N$, $|a_np_n(T_n|\theta)-(2\pi)^{-k/2}I(T_n\in \mathcal T_n)|<\delta$. Then for any $T_n$,  
    \begin{align*}
&\sup_{\sdp} \left| p_n(T_n\,|\,\sdp)-p_\infty^{\mathcal T_n}(T_n\,|\,\sdp)\right|\\
&=\sup_{\sdp} \left| \frac{m(\sdp\,|\,T_n)p_n(T_n)}{\int_{\mathcal T_n} m(\sdp\,|\,T_n=y)p_n(T_n=y)\ d\lambda(y)}-\frac{m(\sdp\,|\,T_n)I(T_n\in \mathcal T_n)}{\int_{\mathcal T_n} m(\sdp\,|\,T_n=y) \ d\lambda(y)}\right|\\
&\leq \sup_{\sdp} m(\sdp\,|\,T_n) \sup_{\sdp} \left|\frac{a_n p_n(T_n)}{\int_{\mathcal T_n} m(\sdp\,|\,T_n=y)a_n p_n(T_n=y)\ d\lambda(y)} - \frac{I(T_n\in \mathcal T_n)}{\int_{\mathcal T_n} m(\sdp\,|\,T_n=y) \ d\lambda(y)}\right|\\
&=L \sup_{\sdp}\max\bigg\{\frac{a_n p_n(T_n)}{\int_{\mathcal T_n} m(\sdp\,|\,T_n=y)a_n p_n(T_n=y)\ d\lambda(y)} - \frac{I(T_n\in \mathcal T_n)}{\int_{\mathcal T_n} m(\sdp\,|\,T_n=y) \ d\lambda(T_n=y)},\\
&\phantom{=L \sup_{\sdp}\max\bigg\{}\frac{I(T_n\in \mathcal T_n)}{\int_{\mathcal T_n} m(\sdp\,|\,T_n=y) \ d\lambda(T_n=y)}-\frac{a_n p_n(T_n)}{\int_{\mathcal T_n} m(\sdp\,|\,T_n=y)a_n p_n(T_n=y)\ d\lambda(y)}\bigg\}\\
&\leq L\sup_{\sdp}\max\bigg\{\frac{((2\pi)^{-k/2}+\delta)I(T_n\in \mathcal T_n)}{\int_{\mathcal T_n} m(\sdp\,|\,T_n=y)\ d\lambda(y)((2\pi)^{-k/2}-\delta)} - \frac{I(T_n\in \mathcal T_n)}{\int_{\mathcal T_n} m(\sdp\,|\,T_n=y)\ d\lambda(y)},\\
&\phantom{\leq L\sup_{\sdp}\max\bigg\{} \frac{I(T_n\in \mathcal T_n)}{\int_{\mathcal T_n} m(\sdp\,|\,T_n=y) \ d\lambda(T_n=y)}-\frac{((2\pi)^{-k/2}-\delta) I(T_n\in \mathcal T_n)}{\int_{\mathcal T_n} m(\sdp\,|\,T_n=y)\ d\lambda(y) ((2\pi)^{-k/2}+\delta)}\bigg\}\\
&=L \frac{2\delta I(T_n\in \mathcal T_n)}{\int_{\mathcal T_n} m(\sdp\,|\,T_n=y)\ d\lambda(y) ((2\pi)^{-k/2} -\delta)}\\
&\leq  \frac{2\delta M}{c((2\pi)^{-k/2}-\delta)},
    \end{align*}
    where by Assumptions \ref{a.location_mechanism} and \ref{a.support_mechanism}, we have that the following two quantities are bounded: $\sup_{\sdp}m(\sdp\,|\,T_n)\leq L<\infty$  and $\int_{\mathcal T_n} m(\sdp\,|\,T_n=y)\ d\lambda(y)\geq \int_{\mathcal T_1}m(\sdp\,|\,T_n=y) \ d\lambda(y)\geq c>0$. Finally, since $\delta>0$ was chosen arbitrarily, and the bound does not depend on $T_n$, we have that 
   \[\sup_{T_n}\sup_{\sdp} \left| p_n(T_n\,|\,\sdp,\theta)-p_\infty^{\mathcal T_n}(T_n\,|\,\sdp)\right|\rightarrow 0,\]
   which implies that $\sup_{\sdp} \mathrm{TV}(p_n(T_n\,|\,\sdp,\theta),p_\infty^{\mathcal T_n}(T_n\,|\,\sdp))\rightarrow 0$, by Scheff\'es Lemma. 

\smallskip

    \noindent \textbf{Step 3}: Still considering both continuous and discrete cases simultaneously, suppose that $\EE t(x_i)=\mu$ which may not be zero. Call $t_i'=t(x_i)-\mu$, $T_n'=T_n-n\mu$, and $\sdp' = \sdp-n\mu$. Since the privacy mechanism is a location family, we have that $m(\sdp'\,|\,T_n')=m(\sdp\,|\,T_n)$ is the density for $\sdp'$ given $T_n'$. Call $p'_n$ the density for $T_n'$. Applying the mean-zero result, we have that 
    \[ \sup_{\sdp'}\mathrm{TV}\left(\frac{m(\sdp'\,|\,T_n')\,p'_n(T_n')}{\int m(\sdp'\,|\,T_n'=y)\,p_n'(T_n'=y)\,d\lambda(y)},\frac{m(\sdp'\,|\,T_n')}{\int m(\sdp'\,|\,T_n'=y)\,d\lambda(y)}\right) \overset p \rightarrow 0.\]
    We can substitute $m(\sdp\,|\,T_n)=m(\sdp'\,|\,T_n')$, $p_n(T_n)=p'_n(T_n')$, and $\int_{\mathcal T_n} m(\sdp\,|\,T_n=y)\,p_n(T_n=y)\,d\lambda(y)=\int_{\mathcal T_n} m(\sdp'\,|\,T_n'=y)p_n'(T_n'=y)\,d\lambda(y)$, which gives the desired result. 

    Finally, for the approximation in the Bayesian setting in 2.,  let $\sdp^{(n)}$ be a sequence of random variables. We write
    \begin{align*}
        \mathrm{TV}\left(p_n(T_n\,|\,\sdp^{(n)}),p_\infty^{\mathcal T_n}(T_n\,|\,\sdp^{(n)})\right)
        &=\frac 12 \int \left| p_n(T_n\,|\,\sdp^{(n)}) - p_\infty^{\mathcal T_n}(T_n\,|\,\sdp^{(n)})\right| \ d\lambda(T_n)\\
        &=\frac 12 \int \left| \EE_{\theta\sim \pi}\left[p_n(T_n\,|\,\sdp^{(n)},\theta) - p_\infty^{\mathcal T_n}(T_n\,|\,\sdp^{(n)})\right]\right| \ d\lambda(T_n)\\
        &\leq \EE_{\theta\sim \pi}\frac 12 \int \left| p_n(T_n\,|\,\sdp^{(n)},\theta) - p_\infty^{\mathcal T_n}(T_n\,|\,\sdp^{(n)})\right| \ d\lambda(T_n)\\
        &=\EE_{\theta\sim \pi} \mathrm{TV}\left(p_n(T_n\,|\,\sdp^{(n)},\theta),p_\infty^{\mathcal T_n}(T_n\,|\,\sdp^{(n)})\right)\\
        &\overset{a.s.}\rightarrow 0,
    \end{align*}
    where the inequality is by Jensen, and we apply the Dominated Convergence Theorem to establish the final convergence, since $\mathrm{TV}\left(p_n(T_n\,|\,\sdp^{(n)},\theta),p_\infty^{\mathcal T_n}(T_n\,|\,\sdp)\right)\overset{a.s.}\rightarrow 0$ for all $\theta$ and total variation is bounded by 1, which is integrable since $\pi$ is a proper prior. 
\end{proof}

\subsection{Proof of Theorem 2 }

We now prove Theorem 2. We first provide a technical lemma which is needed to provide the proofs to our main results. Lemma \ref{lem:condExp} establishes when we can apply a conditional expectation and still have convergence in probability to zero: this will be applied to the Bernstein von-Mises total variation convergence.

\begin{lemma}\label{lem:condExp}
    Suppose that $(X_n,Y_n)$ is a sequence of random variables on some measure space $(\mscr X\times \mscr Y,\mscr F)$. Suppose further that $f:\mscr Y\rightarrow [0,D]$ is a measurable function, where $D>0$ is a fixed constant. Then, if $f(Y_n)\overset p\rightarrow 0$, then $\EE_{Y_n|X_n} f(Y_n)$ converges in mean to zero. % \overset p \rightarrow 0.$$ 
\end{lemma}
\begin{proof}
    Since $|f|\leq D$ and $f(Y_n) \overset p \rightarrow 0$, by the dominated convergence theorem, we have that $\EE_{X_n,Y_n} f(Y_n) \overset p \rightarrow 0$. By the law of total expectation, this is the same as $\EE_{X_n}\left[\EE_{Y_n|X_n} f(Y_n)\right] \overset p \rightarrow 0$. Since $f\geq 0$, this can also be expressed as $\EE_{X_n}\left|\EE_{Y_n|X_n} f(Y_n)-0\right| \overset p \rightarrow 0$, which is the definition of convergence in mean to zero. %Hence the quantity converges in probability to zero.
\end{proof}

We now employ this lemma for the main proof below.

%\thmFullTV*
\begin{proof}[Proof of Theorem 2.]
    We start by decomposing the total variation into several parts by the triangle inequality, and then apply the data processing inequality:
    \begin{align}
\mathrm{TV}&\big(\pi_n(\theta\,|\,T_n)\,p_n(T_n\,|\,\sdp),\pi_\infty(\theta\,|\,T_n)\,p_\infty^{\mathcal T_n}(T_n\,|\,\sdp)\big) \nonumber \\
& \leq \mathrm{TV}\big(\pi_n(\theta\,|\,T_n)\,p_n(T_n\,|\,\sdp),\pi_n(\theta\,|\,T_n)\,p_\infty^{\mathcal T_n}(T_n\,|\,\sdp)\big) \nonumber \\
&\phantom{\leq} +\mathrm{TV}(\pi_n(\theta\,|\,T_n)\,p_\infty^{\mathcal T_n}(T_n\,|\,\sdp),\pi_n(\theta\,|\,T)\,p_n(T_n\,|\,\sdp,\theta_0)\big) \nonumber \\
&\phantom{\leq}+\mathrm{TV}\big(\pi_n(\theta\,|\,T_n)\,p_n(T_n\,|\,\sdp,\theta_0),\pi_\infty(\theta\,|\,T_n)\,p_n(T_n\,|\,\sdp,\theta_0)\big) \nonumber \\
&\phantom{\leq}+\mathrm{TV}\big(\pi_\infty(\theta\,|\,T_n)\,p_n(T_n\,|\,\sdp,\theta_0),\pi_\infty (\theta\,|\,T_n)\,p_\infty^{\mathcal T_n}(T_n\,|\,\sdp)\big) \nonumber \\
&\leq \mathrm{TV}\big(p_n(T_n\,|\,\sdp),p_\infty^{\mathcal T_n}(T_n\,|\,\sdp)\big)\label{eq:1}\\
&\phantom{\leq}+ \mathrm{TV}\big(p_\infty^{\mathcal T_n}(T_n\,|\,\sdp),p_n(T_n\,|\,\sdp,\theta_0)\big)\label{eq:2}\\
&\phantom{\leq}+\mathrm{TV}\big(\pi_n(\theta\,|\,T_n)\,p_n(T_n\,|\,\sdp,\theta_0),\pi_\infty(\theta\,|\,T_n)\,p_n(T_n\,|\,\sdp,\theta_0)\big)\label{eq:3}\\
&\phantom{\leq}+\mathrm{TV}\big(p_n(T_n\,|\,\sdp,\theta_0),p_\infty^{\mathcal T_n}(T_n\,|\,\sdp)\big).\label{eq:4}
\end{align}
The terms in \eqref{eq:1}, \eqref{eq:2} and \eqref{eq:4}all converge almost surely with respect to $\sdp|\theta_0$ by Theorem 1; thus they are also $o_p(1)$ with respect to $\sdp|\theta_0$. For the term in \eqref{eq:3}, we can write it as:
\begin{align*}
    &\mathrm{TV}\big(\pi_n(\theta\,|\,T_n)\,p_n(T_n\,|\,\sdp,\theta_0),\pi_\infty(\theta\,|\,T_n)\,p_n(T_n\,|\,\sdp,\theta_0)\big)\\
    &=\frac 12\int \int \Big|\pi_n(\theta\,|\,T_n=y)\,p_n(T_n=y|\,\sdp,\theta_0)-\pi_\infty(\theta\,|\,T_n=y)\,p_n(T_n=y\,|\,\sdp,\theta_0)\Big|d\mu(\theta) d \lambda(y)\\
    &=\frac 12\int \int \Big|\pi_n(\theta\,|\,T_n=y)-\pi_\infty(\theta\,|\,T_n=y)\Big|d\mu(\theta)  \ p_n(T_n=y\,|\,\sdp,\theta_0)\lambda(y)\\
    &=\EE\limits_{T_n\sim \pi_n(T_n\,|\,\sdp,\theta_0)} \mathrm{TV}\big(\pi_n(\theta\,|\,T_n),\pi_\infty(\theta\,|\,T_n)\big).
\end{align*}
By Assumption \ref{a.bvm}, $\mathrm{TV}\big(\pi_n(\theta\,|\,T_n),\pi_\infty(\theta\,|\,T_n)\big)\overset p\rightarrow 0$ with respect to $T_n\,|\,\theta_0$ as $n\rightarrow \infty$. Since total variation takes values in $[0,1]$, by Lemma \ref{lem:condExp}, we have that 
\[\EE\limits_{T_n\sim p_n(T_n\,|\sdp,\theta_0)} \mathrm{TV}\big(\pi_n(\theta\,|\,T_n),\pi_\infty(\theta\,|\,T_n)\big)\overset p\rightarrow 0,\] 
since convergence in mean implies convergence in probability. 

We have shown that all four terms in \eqref{eq:1}-\eqref{eq:4} are $o_p(1)$. Thus, the result holds. 
\end{proof}

\begin{comment}
\subsection{MCMC Method to Sample \texorpdfstring{$T_n$}{T_n}}\label{s:mcmc}
We propose the following MCMC strategy to sample $t(x)|\sdp$:
\begin{enumerate}
    \item Sample $\theta\sim \pi(\theta)$
    \item Sample $t(x_i')\,|\,\theta \sim p(t(x_i')\,|\,\theta)$
    \item Given a previous $t(x_1),\ldots, t(x_n)$, swap $t(x_i)$ for $t(x_i')$ with probability
    \[ \min\left\{1,\frac{m(\sdp\,|\,T'_n)}{m(\sdp\,|\,T_n)}\right\}\]
\end{enumerate}
Note that if $m$ is an $\epsilon$-DP mechanism, then the acceptance probability is bounded below by $\exp(-\ep)$. This ensure that we are able to effectively propose and accept new values. This method is similar to the data augmentation MCMC approach of \citet{ju2022data}. Once we have run the MCMC sufficiently long, we can sample $\theta\,|\,T_n$ using standard methods, or use the large sample approximation proposed in the manuscript. 
\end{comment}
\section{Normal Approximate Posterior Sampling}
\label{s.totcov_approx}

Modifications ca be made to the proposed sampling algorithm in Section \ref{s:sample} as mentioned therein. Indeed, since we are mainly interested in the posterior mean and covariance for common inferential tasks, we can use the law of total covariance decompose the posterior covariance:
%Instead of sampling both $T|\sdp$ and $\theta|T$, we can observe that 
\[\text{cov}(\theta\,|\,\sdp)=\EE_{T_n\sim \sdp} \text{cov}(\theta\,|\,T_n) + \text{cov}_{T_n\sim \sdp}(\EE (\theta\,|\,T_n)),\]
where $\text{cov}(\theta\,|\,T_n)$ and $\EE(\theta\,|\,T_n)$ are the posterior covariance and posterior mean of $\theta\,|\,T_n$. If we have $R$ i.i.d. samples $T^{(1)}_n,\ldots, T^{(R)}_n\sim p(T_n\,|\,\sdp)$, we can approximate $\text{cov}(\theta\,|\,\sdp)$ using 
\[\widehat{\text{cov}}(\theta|\sdp) = \frac 1R \sum_{i=1}^R\widehat{\text{cov}}(\hat\theta(T^{(i)}_n))+\frac{1}{R-1}\sum_{i=1}^R\left(\hat\theta(T^{(i)}_n)-\overline{\theta}\right)\left(\hat\theta(T^{(i)}_n)-\overline {\theta}\right)^\top,\]
where $\hat\theta(T^{(i)}_n)$ is an efficient estimator, such as the MLE, $\bar\theta$ is the  mean of $\hat\theta(T^{(1)}_n),\ldots,\hat\theta(T^{(R)}_n)$ and $\widehat{\text{cov}}(\hat\theta(T^{(i)}_n))$ is an estimator for the covariance of $\hat\theta(T^{(i)}_n)$. If we expect the posterior distribution $\theta\,|\,\sdp$ to be approximately normal, then we can use $\hat\theta$ and $\widehat{\text{cov}}(\hat\theta)$ to produce credible regions. A modified version of Algorithm \ref{algo.sample} based on this approach is given in Algorithm \ref{algo.totcov}. The benefit of this approach is that we no longer need to sample $\theta\,|\,T_n$, avoiding the computational cost and additional Monte Carlo error. On the other hand, if the approximate posterior distribution is not approximately normally distributed, then credible regions should be estimated by sampling both $T^{(i)}_n\,|\,\sdp$ and $\theta_i\,|\,T^{(i)}_n$, instead of through the estimated covariance matrix.

\vspace{0.1em}
\setcounter{algorithm}{1}

\begin{algorithm}
\footnotesize % Smaller font to reduce vertical space
\caption{PUMBA (Mean and Covariance)}
\begin{algorithmic}[1]
\setlength{\itemsep}{1pt} % Reduce line spacing
\REQUIRE Private output $\sdp$; Privacy mechanism $m$; Efficient estimator $\hat\theta(T_n)$; Number of Monte Carlo samples $R$; Support $\mathcal{T}_n$; Sample size $n$.
\ENSURE  $\bar{\theta}$ and $\widehat{\text{cov}}(\hat{\theta})$ as the approximate posterior mean and covariance.
\STATE Sample $T^{(1)}_n,\ldots, T^{(R)}_n \overset{\text{iid}}{\sim} p_\infty^{\mathcal T_n}(T_n\,|\,\sdp)$ 
\FOR{each $T^{(i)}_n$}
    \STATE Compute $\hat{\theta}(T^{(i)}_n)$.
    \STATE Compute $\widehat{\text{cov}}(\hat{\theta}(T^{(i)}_n))$ (the estimated covariance of $\hat{\theta}(T^{(i)}_n)$).
\ENDFOR
\STATE Compute parameter estimate:
\[
\bar{\theta} = \frac{1}{R} \sum_{i=1}^{R} \hat{\theta}(T^{(i)}_n).
\]
\STATE Compute estimated covariance:
\[
\widehat{\text{cov}}(\hat{\theta}) = \frac{1}{R} \sum_{i=1}^{R} \widehat{\text{cov}}(\hat{\theta}(T^{(i)}_n))
+ \frac{1}{R-1} \sum_{i=1}^{R} \left(\hat{\theta}(T^{(i)}_n)-\bar{\theta}\right)\left(\hat{\theta}(T^{(i)}_n)-\bar{\theta}\right)^\top.
\]
\end{algorithmic}
\label{algo.totcov}
\end{algorithm}

\section{Data Augmentation MCMC Details for Section 4.2}\label{s:linDetails}

Our Bayesian model used in the data augmentation MCMC of Section 4.2 follows that of \citet{awan2024statistical}: we model the original data as $y_i|x_i,\beta,\tau \sim N((1,x_i)\beta,\tau^{-1}I)$ and $x_i|\mu,\Phi \sim N_p(\mu,\Phi^{-1})$. We use priors $\beta|\tau \sim N_{p+1}(m,\tau^{-1}V^{-1})$, $\tau \sim \mathrm{Gamma}(a/2,b/2)$ (using the shape-rate parametrization), $\Phi\sim \mathrm{Wishart}_p(d,W)$, and $\mu\sim N(\theta,\Sigma)$.  In our setting, $p=1$. 
\begin{comment}
\begin{align*}
\mu &\sim N(\theta, \Sigma)\\
\Phi &\sim  \mathrm{Wishart}_p(d,W)\\
\beta | \tau &\sim  N_{p+1}(m,\tau^{-1}V^{-1})\\
\tau &\sim \mathrm{Gamma}(a/2, b/2), \text{using the shape, rate parameterization}\\
X_i | \mu, \phi &\sim N_p(\mu, \phi^{-1})\\
Y_i | X_i, \beta, \tau &\sim  N(X_i\beta, \tau^{-1} I)\\
\end{align*}
\end{comment}
In total, the parameters are $(\beta, \tau, \mu, \Phi)$,  
the hyperparameters are $(m, V, a, b, \theta, \Sigma, d, W)$, where $X$ is a matrix whose rows are $(1,x_i)$ and $Y$ is a vector whose entries are the $y_i$'s. Since the $x$'s and $y$'s are assumed to lie in $[0,1]$, we do not incorporate clamping into the mechanism density. 

\section{Literature Review for Drivers of Homeownership}
\label{s:lit_rev_home}

Classic research consistently demonstrated that minorities and low-income households faced pronounced barriers, and demographic and socio-economic factors, housing market, and broader economic policies critically influenced individuals' ability to own homes, contributing to persistent disparities in homeownership rates \citep{kuebler2013new, grinstein2011homeownership, segal1998trends, goodman2018homeownership}. Racial disparities have long been recognized as significant barriers to homeownership. Systemic inequalities and discriminatory lending practices significantly limited homeownership opportunities for minority households, perpetuating racial wealth gaps \citep{gyourko1999analyzing}. Another perspective influencing home ownership focused on financial stability and affordability. Unemployment significantly affects homeownership rates by undermining households’ financial stability and their ability to qualify for mortgages \citep{blanchflower2013does, gathergood2011racial}. In addition, unexpected financial shock such as medical expenses will also have mediate impact on financial stability. While health coverage is less frequently studied in direct relation to homeownership, it indirectly impacts housing stability by influencing household finances \citep{himmelstein2005illness, schoen2005taking}. For example, \cite{himmelstein2005illness} demonstrated that medical debt was a leading cause of financial distress and foreclosure, making health coverage a critical factor for maintaining homeownership. The spending power and cost of living are major perspectives in determining affordability that impacts homeownership accessibility. Income and poverty are fundamental measures of spending power influencing home affordability and access to credit \citep{henderson1983model, rohe1994effects}. Households with higher income levels have higher homeownership rates and the intention to house tenure, as income directly influences the ability to save for a down payment, qualify for a mortgage, and manage housing costs \citep{drew2014believing}. Poverty also disproportionately affects minority households, further entrenching racial disparities in homeownership rates \citep{rohe1994impact, segal1998trends}. In terms of population density, \cite{ross2008neighborhood} noted that high-density urban areas often had distinct housing markets, where elevated property prices and limited supply reduce overall homeownership rates. Population density is frequently controlled in study of socio-economic determinants of homeownership. By accounting for differences in urban, suburban, and rural housing markets, population density helps isolate the effects of key factors such as race, income, and housing costs. Population density was also controlled to distinguish the direct effects of racial inequality and economic challenges on homeownership \cite{blanchflower2013does, desilva2012housing}.

\section{ACS Study with Logistic Regression}
\label{s.acs_study_logistic}

We now consider an alternative model to study the drivers of homeownership which consists in logistic regression where the response (which is a proportion of homeownership within each county) is weighted by the total population in the counties. In this way counties with larger populations get weighted more compared to those with lower populations.

As for linear regression, we perform a similar simulation study to that in Section \ref{s.home_sim} based on the coefficient estimates from the non-private data. The results of this study (using same parameters as before) can be found in Table \ref{tab:coverage_results_logistic} where we can see that all methods are close the desired nominal covarage of $1 - \alpha = 0.95$. The Naive seems to slightly undercover the coefficients (under $H_0$ and $H_A$) for Indigenous and Housecost variables, but otherwise is in line with the non-private version, while the proposed PUMBA is close to the nominal level or generally above it (thereby reducing type-I errors).

\begin{table}[t]
    \centering
    \footnotesize
    \renewcommand{\arraystretch}{1.0}
    \setlength{\tabcolsep}{6pt}
    \begin{threeparttable}
    \caption{Coverage probabilities for Non-DP, Naive, and PUMBA under the null and alternative hypotheses using logistic regression (weighted with county population). The number of simulations is $B = 1{,}000$, and the expected nominal level is $1 - \alpha = 0.95$.}
    \label{tab:coverage_results_logistic}
    \begin{tabular}{lcccccc}
        \toprule
         & \multicolumn{3}{c}{$H_0$} & \multicolumn{3}{c}{$H_A$} \\
         & Non-DP & Naive & PUMBA & Non-DP & Naive & PUMBA \\
        \midrule
        Intercept    & 0.946 & 0.933 & 0.954 & 0.959 & 0.944 & 0.978 \\
        Black        & 0.961 & 0.959 & 0.961 & 0.946 & 0.944 & 0.949 \\
        Asian        & 0.960 & 0.962 & 0.964 & 0.938 & 0.935 & 0.954 \\
        Indigenous   & 0.954 & 0.923 & 0.948 & 0.943 & 0.908 & 0.933 \\
        Other        & 0.948 & 0.948 & 0.950 & 0.960 & 0.953 & 0.967 \\
        Unemployed   & 0.948 & 0.940 & 0.953 & 0.957 & 0.947 & 0.957 \\
        Poverty      & 0.952 & 0.950 & 0.958 & 0.955 & 0.942 & 0.973 \\
        No Insurance & 0.954 & 0.952 & 0.958 & 0.974 & 0.967 & 0.987 \\
        Housecost    & 0.947 & 0.939 & 0.949 & 0.948 & 0.924 & 0.976 \\
        Pop. Density & 0.953 & 0.951 & 0.952 & 0.945 & 0.944 & 0.960 \\
        \bottomrule
    \end{tabular}
    \end{threeparttable}
\end{table}

Having confirmed the good behavior of the three methods, we now compare them when applied directly to the data. The results are presented in Table \ref{tab:regression_results_logistic} where, compared to the analysis in Section \ref{s:homeownership}, we now see a change of sign for the Indigenous variable. Overall, all three methods agree on coefficients estimates and significance of all variables. 

A possible explanation for the sign change for the Indigenous variable is as follows: The main difference between the two models is the weighting. In the linear regression, all counties are given equal weight, whereas the logistic has counties weighted by their population. We suspect that there are some small counties with predominantly indigenous residents, with low home ownership. Because of these counties, the linear regression identifies that if a county has a high proportion of indigenous residents, it may be similar to these small counties. On the other hand, the logistic regression puts much smaller weight on small counties. This model says that (mostly in large counties), having a higher indigenous population has a very small impact on home ownership (compared to the other racial groups). The choice of which model is more appropriate may depend on the context.

\begin{table}[t]
    \centering
    \footnotesize % Reduce font size to save space
    \renewcommand{\arraystretch}{1.0} % Tighter row spacing
    \setlength{\tabcolsep}{6pt} % Tighter column spacing
    \begin{threeparttable}
    \caption{Coefficient Estimates and corresponding p-values (in parentheses) for Non-DP approach (first column); Naive approach (second column) and PUMBA (third column). Significance levels: * $< 0.05$, ** $< 0.01$, *** $< 0.001$.}
    \label{tab:regression_results_logistic}
    \begin{tabular}{l
        S[table-format=1.3]@{\,}c
        S[table-format=1.3]@{\,}c
        S[table-format=1.3]@{\,}c}
        \toprule
         & \multicolumn{2}{c}{Non-DP} & \multicolumn{2}{c}{Naive} & \multicolumn{2}{c}{PUMBA} \\
        & {Estimate} & {(p-value)} & {Estimate} & {(p-value)} & {Estimate} & {(p-value)} \\
        \midrule
        Intercept  & 1.119 & ($<$0.001) *** & 1.118 & ($<$0.001) *** & 1.120 & ($<$0.001) *** \\
        Black      & -0.925 & ($<$0.001) *** & -0.925 & ($<$0.001) *** & -0.924 & ($<$0.001) *** \\
        Asian      & -2.631 & ($<$0.001) *** & -2.630 & ($<$0.001) *** & -2.630 & ($<$0.001) *** \\
        Indigenous & 0.021 & (0.001) ** & 0.024 & ($<$0.001) *** & 0.031 & ($<$0.001) *** \\
        Other      & -1.491 & ($<$0.001) *** & -1.493 & ($<$0.001) *** & -1.490 & ($<$0.001) *** \\
        Unemployed & -0.692 & ($<$0.001) *** & -0.677 & ($<$0.001) *** & -0.679 & ($<$0.001) *** \\
        Poverty    & -2.934 & ($<$0.001) *** & -2.932 & ($<$0.001) *** & -2.930 & ($<$0.001) *** \\
        No Insurance & -0.131 & ($<$0.001) *** & -0.127 & ($<$0.001) *** & -0.128 & ($<$0.001) *** \\
        Housecost  & 0.111 & ($<$0.001) *** & 0.111 & ($<$0.001) *** & 0.111 & ($<$0.001) *** \\
        Pop. Density & \num{-2.376e-05} & ($<$0.001) *** & \num{-2.376e-05} & ($<$0.001) *** & \num{-2.37e-05} & ($<$0.001) *** \\
        \bottomrule
    \end{tabular}
    \end{threeparttable}
\end{table}

\bibliographystyle{chicago}

\bibliography{bibliography}
\end{document}